\documentclass[a4paper]{cas-sc}

\usepackage[authoryear,longnamesfirst]{natbib}

\def\tsc#1{\csdef{#1}{\textsc{\lowercase{#1}}\xspace}}
\tsc{WGM}
\tsc{QE}

\usepackage{amssymb}
\usepackage{amsmath}
\usepackage{amsthm}
\usepackage{float}
\usepackage{array}
\usepackage{comment}
\usepackage[noend]{algpseudocode}
\usepackage{algorithm,algorithmicx}


\usepackage{multirow}
\usepackage{tabularx}
\usepackage{booktabs}
\usepackage{makecell}
\usepackage{subfigure}
\usepackage{threeparttable}
\hypersetup{
    colorlinks=true,
    linkcolor=blue,
    filecolor=blue,      
    urlcolor=blue,
    citecolor=cyan,
}

\usepackage{soul}  
\sethlcolor{yellow}

\usepackage{threeparttable}

\DeclareMathOperator*{\minimize}{minimize}
\usepackage{soul}

\usepackage{lscape}
\usepackage{enumitem}
\usepackage{tabularx}
\usepackage{amsthm}
\newtheorem{proposition}{Proposition}
\newtheorem{remark}{Remark}
\usepackage{graphicx,subfigure}
\usepackage{hyperref}

\usepackage{multirow}
\usepackage{lscape}
\usepackage{longtable}
\usepackage{makecell} 
\usepackage{tablefootnote} 

\begin{document}
\let\WriteBookmarks\relax
\def\floatpagepagefraction{1}
\def\textpagefraction{.001}

\shorttitle{}    

\shortauthors{Fan and Gu}  

\title [mode = title]{Fly-by transit: A novel door-to-door shared mobility with minimal stops}  



\author[1]{Wenbo Fan} 
\ead{wenbo.fan@polyu.edu.hk}
\cormark[1]

\author[1,2]{Weihua Gu}
\ead{weihua.gu@polyu.edu.hk}

\affiliation[1]{organization={Department of Electrical and Electronic Engineering, The Hong Kong Polytechnic University},
	            addressline={Kowloon},
	            city={Hong Kong},
	             country={China}}
\affiliation[2]{organization={Otto Poon Charitable Foundation Smart Cities Research Institute, The Hong Kong Polytechnic University},
	addressline={Kowloon},
	city={Hong Kong},
	country={China}}


\cortext[1]{Corresponding author}

\begin{abstract}
This paper introduces fly-by transit (FBT), a novel mobility system that employs modular mini-electric vehicles (mini-EVs) to provide door-to-door shared mobility with minimal stops. Unlike existing modular minibus concepts that rely on in-motion coupling and passenger transfers---technologies unlikely to mature soon---FBT lowers the technological barriers by building upon near-term feasible solutions. The system comprises two complementary mini-EV modules: low-cost \textit{trailers} for on-demand feeder trips and high-performance \textit{leaders} that guide coupled trailers in high-speed platoons along trunk lines. Trailers operate independently for detour-free feeder services, while \textit{stationary coupling} at designated hubs enables platoons to achieve economies of scale (EoS). \textit{In-motion decoupling of the tail trailer} allows stop-less operation without delaying the main convoy.

As a proof of concept, a stylized corridor model is developed to analyze optimal FBT design. Results indicate that FBT can substantially reduce travel times relative to conventional buses and lower operating costs compared with e-hailing taxis. Numerical analyses further demonstrate that FBT achieves stronger EoS than both buses and taxis, yielding more than 13\% savings in generalized system costs. By addressing key limitations of existing transit systems, this study establishes FBT as a practical and scalable pathway toward transformative urban mobility and outlines directions for future research.
\end{abstract}




\maketitle

\section{Introduction}
\subsection{Background}
Traditional transportation systems have faced persistent challenges in balancing three core properties: accessibility, mobility, and cost-effectiveness (AMC). Accessibility refers to the ubiquitous availability of the service, mobility to the efficiency of movement, and cost-effectiveness to its economic viability. 
Achieving optimal performance across all three dimensions has proven exceedingly difficult for a single transportation mode—whether individual (e.g., taxis, private cars) or collective (e.g., buses, rail transit). This inherent trade-off can be referred to as the AMC trilemma or the “impossible trinity,” as illustrated in Fig. \ref{fig_trinity}, where existing travel modes typically excel in two of these dimensions at the expense of the third. 

To address this trilemma, traditional efforts have focused on integrating different modes into hierarchical systems, such as combining buses or taxis as feeders to rail transit \citep[e.g.,][]{fan2018optimal, Luo2019a}. The rationale is that by leveraging the complementary strengths of each mode, a more balanced AMC profile might be achieved. Nevertheless, in practice, such integration often faces significant hurdles: the operational complexity of synchronizing diverse services, the misalignment of institutional objectives, and the challenge of delivering seamless patron experiences across multiple providers, among others.
\begin{figure}[!ht]
	\centering
	\includegraphics[width=0.5\linewidth]{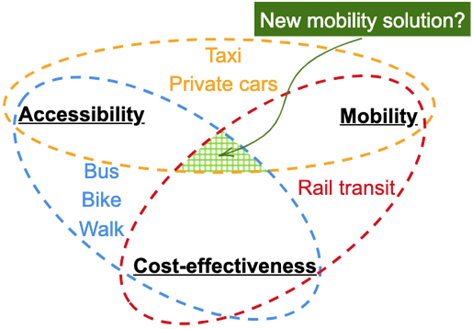}
	\caption{ Impossible trinity of urban mobilities.}
	\label{fig_trinity}
\end{figure}

Recent rapid advances in two transformative technologies—autonomous driving and modular vehicles—present new opportunities to revolutionize urban mobility and revisit the longstanding challenges posed by the AMC trilemma. Unlike previous approaches that primarily sought to integrate existing multi-entity systems, these technological innovations offer the potential for more profound solutions within a single, unified system. Inspired by these developments, this research seeks to investigate a novel modular mobility solution that can inherently balance the three AMC dimensions. 

\subsection{Research motivation and contribution}

The integration of autonomous driving with modular vehicles has led to the emergence of modular autonomous vehicles (MAVs)---also referred to as modular autonomous buses (MAB) when adopt minibuses as vehicle modules---flexible units that can operate individually or couple into larger platoons for higher capacity \citep{qu2021communications}. This concept has the potential to simultaneously enhance accessibility and mobility while exerting only modest pressure on cost-effectiveness \citep{lin2022autonomous,luo2025optimal}. As a result, MAVs have stimulated a surge in research on novel operational strategies and new paradigms for modular transit services \citep{chen2020operational, chen2021designing, cheng2024autonomous, hannoun2022modular, khan2025no}. Prototype MAVs have been manufactured and tested in Europe \citep{boldrini2017stackable}, Dubai \citep{gecchelin2022selectively}, Beijing \citep{he2025lidar,lin2025big}, and New Zealand \citep{ye2024exploring}.

Nevertheless, current MAV studies have been criticized for relying on overly optimistic assumptions regarding future technologies—particularly in-motion coupling and patron transfers between moving vehicles in a dynamic, mixed-traffic urban environment. These features, while conceptually promising, still lack robust theoretical foundations and practical validation, raising concerns about their feasibility in the near term \citep{he2025lidar, luo2025optimal}.

In response, this paper proposes a new modular mobility solution, termed “fly-by transit” (FBT), designed around technologies that are feasible in the short to medium term. We aim to demonstrate that FBT systems hold significant promise of striking a better balance in the longstanding AMC trilemma and filling the gap in urban mobility, as illustrated in Fig. \ref{fig_trinity}. Our contributions include multiple innovations in both vehicle design and operational strategies, enabling rapid door-to-door transit with minimal stops. As a proof of concept, we develop optimal design models for FBT corridors, showcasing significant Economies of Scale (EoS) in terms of reduced average system costs as demand increases. Building on the promising results, the paper proposes a research agenda to advance foundational theories and algorithms that optimize FBT designs and operations in general networks, thereby facilitating the real-world implementation of these designs.

The remainder of this paper is organized as follows. The next section reviews literature on modular mobility to position FBT in comparison to existing MAV systems. Section \ref{sec_concept} conceptualizes FBT from the perspectives of technological requirements, operational characteristics, and patron experiences. Section \ref{sec_proof_of_concept} provides a proof of concept by modeling an FBT corridor design problem and demonstrating FBT’s superior performance compared to traditional taxi and bus systems. Section \ref{sec_agenda} further proposes a research agenda for the incubation of FBT. The final section concludes the paper.

\section{Literature review}
We classify the existing literature on MAVs into four categories, as outlined in Table \ref{tab_category}: (1) conceptual development, (2) operations with variable capacity, (3) operations with in-motion transfer (IMT), and (4) microscopic dynamics. To maintain a focused analysis, our review excludes studies that do not explicitly leverage modularity during operation \citep[e.g.,][]{ding2024exploring, ng_semi--demand_2024, tu_-demand_2024}. We also omit alternative modular vehicle concepts outside the road transportation domain, such as those in railways and military \citep[e.g.,][]{pei2023robust, li2020ai}, as well as designs limited to modular chassis with unpowered cabins that lack independent operability \citep[e.g.,][]{shen2023dynamic, friedrich2019new}. 
\begin{table}[!ht]
	\centering
	\caption{Literature on Modular Autonomous Vehicles (MAVs).}
	\label{tab_category}
	\resizebox{\columnwidth}{!}{%
		\begin{tabular}{lll}
			\hline
			\multicolumn{1}{|l|}{\textbf{Category}} &
			\multicolumn{1}{l|}{\textbf{Work}} &
			\multicolumn{1}{l|}{\textbf{Key feature}} \\ \hline
			\multicolumn{1}{|l|}{\multirow{6}{*}{MAV concepts}} &
			\multicolumn{1}{l|}{\begin{tabular}[c]{@{}l@{}} \cite{boldrini2017stackable,gobillot2018esprit} \\ \cite{iacobucci2022multi} \end{tabular}} &
			\multicolumn{1}{l|}{EU-MSC$^*$: single mini-EV module} \\ \cline{2-3} 
			\multicolumn{1}{|l|}{} &
			\multicolumn{1}{l|}{\cite{gecchelin2022selectively,gecchelin2017connectable}} &
			\multicolumn{1}{l|}{Dubai-MAB: single minibus module} \\ \cline{2-3} 
			\multicolumn{1}{|l|}{} &
			\multicolumn{1}{l|}{\cite{lin2025big}} &
			\multicolumn{1}{l|}{Beijing-MAB: single minibus module} \\ \cline{2-3} 
			\multicolumn{1}{|l|}{} &
			\multicolumn{1}{l|}{\cite{liu2018deployment,rau2019dynamic}} &
			\multicolumn{1}{l|}{Singapore-MAB: single minibus module} \\ \cline{2-3} 
			\multicolumn{1}{|l|}{} &
			\multicolumn{1}{l|}{\cite{luo2025optimal,zhang2020modular}} &
			\multicolumn{1}{l|}{Variants of Dubai-MAB} \\ \cline{2-3} 
			\multicolumn{1}{|l|}{} &
			\multicolumn{1}{l|}{\textbf{This paper}} &
			\multicolumn{1}{l|}{FBT: two mini-EV modules} \\ \hline
			\multicolumn{1}{|l|}{\multirow{6}{*}{Variable capacity}} &
			\multicolumn{1}{l|}{\cite{chen2020operational,chen2019operational}} &
			\multicolumn{1}{l|}{Shuttle service} \\ \cline{2-3} 
			\multicolumn{1}{|l|}{} &
			\multicolumn{1}{l|}{\begin{tabular}[c]{@{}l@{}}\cite{chen2022continuous,chen2021designing,dai2020joint} \\ \cite{ji2021scheduling, liu2020using,liu2023integrated,shi2021operations} \\ \cite{tian2023joint, wang2023optimal,zhang2024optimising}\end{tabular}} &
			\multicolumn{1}{l|}{Single line} \\ \cline{2-3} 
			\multicolumn{1}{|l|}{} &
			\multicolumn{1}{l|}{\cite{guo2018stochastic,liu2021improving,tang2024optimisation}} &
			\multicolumn{1}{l|}{(Semi-)flexible route service} \\ \cline{2-3} 
			\multicolumn{1}{|l|}{} &
			\multicolumn{1}{l|}{\cite{gong2021transfer,guo2023modular}} &
			\multicolumn{1}{l|}{Customized bus service} \\ \cline{2-3} 
			\multicolumn{1}{|l|}{} &
			\multicolumn{1}{l|}{\cite{dakic2021design,pei2021vehicle,wang2024multimodal}} &
			\multicolumn{1}{l|}{Network} \\ \cline{2-3} 
			\multicolumn{1}{|l|}{} &
			\multicolumn{1}{l|}{\textbf{This paper}} &
			\multicolumn{1}{l|}{Corridors, networks, \& feeders} \\ \hline
			\multicolumn{1}{|l|}{\multirow{3}{*}{In-motion transfer}} &
			\multicolumn{1}{l|}{\begin{tabular}[c]{@{}l@{}} \cite{lin2024bunching,liu2024alleviating,khan2023application} \\ \cite{khan2025no,khan2023bus} \\ \cite{romea2021analysis,tian2022planning,zou2024operational} \end{tabular}} &
			\multicolumn{1}{l|}{Single line} \\ \cline{2-3} 
			\multicolumn{1}{|l|}{} &
			\multicolumn{1}{l|}{\begin{tabular}[c]{@{}l@{}}\cite{caros2021day,cheng2024autonomous} \\ \cite{fu2023dial, khan2024seamless} \\ \cite{zermasli2023feeder,zhang2020modular} \end{tabular}} &
			\multicolumn{1}{l|}{Network} \\ \cline{2-3} 
			\multicolumn{1}{|l|}{} &
			\multicolumn{1}{l|}{\textbf{This paper}} &
			\multicolumn{1}{l|}{Systems without IMT} \\ \hline
			\multicolumn{1}{|l|}{\multirow{3}{*}{\begin{tabular}[c]{@{}l@{}}Microscopic \\ dynamics\end{tabular}}} &
			\multicolumn{1}{l|}{ \begin{tabular}[c]{@{}l@{}} \cite{han2024planning,he2025lidar,li2022trajectory} \\ \cite{ma2025safety, ye2024modular} \end{tabular}} &
			\multicolumn{1}{l|}{\begin{tabular}[c]{@{}l@{}} In-motion coupling; \\ Inter-vehicle passage \end{tabular} } \\ \cline{2-3} 
			\multicolumn{1}{|l|}{} &
			\multicolumn{1}{l|}{\cite{li2023trajectory,ye2024modular}} &
			\multicolumn{1}{l|}{In-motion decoupling} \\ \cline{2-3} 
			\multicolumn{1}{|l|}{} &
			\multicolumn{1}{l|}{\textbf{This paper}} &
			\multicolumn{1}{l|}{\begin{tabular}[c]{@{}l@{}}Stationary coupling; \\ In-motion decoupling tail trailer \end{tabular}} \\ \hline
			\multicolumn{3}{l}{$^*$ Autonomous driving is not a must for the original concept \citep{gobillot2018esprit}.}
		\end{tabular}%
	}
\end{table}

\subsection{Conceptual development of MAVs}
For ease of reference, we designate existing MAV concepts by the names of the cities in which they were originally proposed or tested. In addition to the proposed FBT, we have identified four prominent MAV concepts in the literature: the EU-MSC, Dubai-MAB, Beijing-MAB, and Singapore-MAB, where MSC stands for Modular Shared Cars, which belong to the mini-electric vehicle (mini-EV) category and differ from the minibuses in the MAB concepts. Among these, EU-MSC and Dubai-MAB have been prototyped and publicly demonstrated, showcasing the functionality of key components, such as the coupling/decoupling mechanism, propulsion system, battery management system, and inter-vehicle communication system. The detailed characteristics of these concepts are compared in Table \ref{tab_bigtable}. The following discussion synthesizes their key similarities and differences.

EU-MSC, also known as ESPRIT or stackable cars, is designed to accommodate one to three seated patrons with the capability of stationary coupling and decoupling \citep{boldrini2017stackable, gobillot2018esprit, iacobucci2022multi}. The primary objective of EU-MSC is to improve fleet rebalancing efficiency in carsharing services. By coupling multiple unoccupied vehicles, a single relocator can drive and relocate the whole platoon simultaneously instead of repositioning each vehicle individually as in traditional carsharing systems. Specifically, one dedicated relocator can operate a platoon of up to eight vehicles, while a customer can drive one vehicle coupled with one vacant unit \citep{boldrini2017relocation}.

Dubai-MAB and Beijing-MAB both utilize modular minibuses for seated patrons and standees but differ in capacity: the former holds 4--10 patrons per unit \citep{gecchelin2017connectable,gecchelin2022selectively}, while the latter accommodates 10--25 patrons \citep{lin2025big}. Both concepts rely on fully autonomous driving, in-motion coupling and decoupling, and inter-vehicle passage across physically platooned MAVs to facilitate IMT. In contrast, Singapore-MAB (also called DART) employs virtual platooning via wireless vehicle-to-vehicle (V2V) communication and cooperative adaptive cruise control, which precludes IMT between vehicles \citep{liu2018deployment, rau2019dynamic}. There are also some variant concepts in \citep{zhang2022efficiency} that utilize semi-autonomous minibuses (with an operator only in the lead vehicle of the platoon) and in \citep{luo2025optimal} that exclude in-motion coupling and transfers. 

\subsection{Operations with variable capacity}
All the above MAV concepts share the property of variable capacity---a feature that is unprecedented in traditional fixed vehicle or bus models. This capability allows MAV-based services to flexibly accommodate fluctuating demand across both temporal and spatial dimensions. The variable capacity of FBT is manifested in the adaptability of platoon lengths, tailored to the specific demand profiles of different lines and around stations. Such adaptability in capacity not only enhances transport efficiency but also significantly improves operational cost-effectiveness.

Existing research on this topic spans a broad range of transit service models, as summarized in Table \ref{tab_category}, including one-to-one shuttle services \citep[e.g.,][]{chen2019operational}, fixed-route lines \citep[e.g.,][]{cao_optimizing_2025}, multi-line networks \citep[e.g.,][]{wang2024multimodal}, flexible-route bus services \citep[e.g.,][]{tang2024optimisation}, and on-demand feeder systems; see \citep{luo2025optimal} for a recent review. 

Nevertheless, while variable capacity in existing MAVs allows for flexible adaptation to temporal-spatial demand variations, it does not address the challenge of serving patrons' diverse origins and destinations (ODs), which necessitates frequent intermediate stops for boarding and alighting as in traditional transit systems. This limitation is expected to be addressed by IMT, as discussed below. 


\subsection{Operations with in-motion transfer} \label{subsec_IMT}
The IMT is introduced to redistribute onboard patrons among moving vehicles, ensuring that those with similar destinations share the same vehicle traveling in the same direction or stopping at the same station. It is anticipated to reduce or even eliminate patrons’ out-of-vehicle transfers and the associated delays, a persistent challenge in conventional bus systems and Singapore-MAB systems (which lack IMT capability due to the virtual coupling). 

Most existing research on IMT has focused on innovative operational strategies, often assuming that IMT can occur anywhere, at any time, and any number of times \citep[e.g.,][]{zhang2020modular,zermasli2023feeder,cheng2024autonomous}. 
A limited number of studies have introduced constraints to mitigate IMT's negative impacts on patron experience and general traffic, such as limiting the number of IMTs \citep{caros2021day}, restricting IMT to specific, designated locations \citep{tian2022planning}, and minimizing patrons’ total number of IMTs and MAV’s decoupling and recoupling maneuvers \citep{wu2021modular}. 

Notably, IMT has been leveraged to enable stop-less transit services, as demonstrated in the works of \citep{khan2025no,khan2024seamless,romea2021analysis}. These studies represent the first application of the MAV concept to minimize intermediate stop delays for boarding and alighting---a factor that accounts for more than 30\% of total bus travel time \citep{walker2011human}. Their proposed system, termed the Stop-Less Autonomous Modular (SLAM) bus, involves a main convoy that passes all stops without halting, while only the tail pod detaches and stops to allow patrons to alight. Patrons intending to alight must move in advance to the designated detaching pod. For boarding, patrons arrive at the stop and wait in a standby MAV, which will be dispatched to couple in motion with the passing platoon. Consequently, the commercial speed of the SLAM buses is competitive with that of taxis and private cars.

Nonetheless, the SLAM bus system still requires patrons to walk to and from stops to access or egress the service, which reduces the overall door-to-door trip speed. Thus, while IMT-based systems offer substantial improvements in line-haul efficiency, further innovations are needed to optimize the complete journey of patrons. 

\subsection{Microscopic dynamics of coupling and decoupling}
Although in-motion coupling forms the technical basis of most MAV studies, it is not a premise of FBT. It is recently noticed that achieving reliable in-motion coupling is substantially more complex than virtual platooning, as it requires highly precise motion/attitude control of multiple vehicles at high cruising speeds \citep{he2025lidar, ma2025safety}. To date, this challenge has only been partially addressed in a handful of pioneering studies \citep{han2024planning,he2025lidar,li2022trajectory, ma2025safety}, which proposed trajectory planning solutions that rely on several idealized assumptions and remain far from practical on public roads. Additionally, \citet{ye2024modular} have also examined the broader impacts of MAV coupling and decoupling operations on traffic dynamics in mixed traffic environments.

Insights from other fields further illuminate these challenges. For example, research on rail systems \citep{nold2021dynamic} highlights substantial technical barriers to in-motion coupling, even for track-guided trains that are not subject to general traffic or signalized intersections. These systems must reconcile speed differentials while minimizing collision energy, underscoring the inherent complexity of the topic. Although in-motion docking has been studied and implemented in spacecraft \citep{wang2023brief}, the space environment differs fundamentally from urban traffic and does not provide directly transferable solutions.

\subsection{Research gaps}
The review highlights that most existing MAV concepts adopt a minibus-like vehicle module with a capacity of over four seated and standees. These systems inherit the challenges associated with patrons' diverse OD pairs that may necessitate en-route detours and stops, similar to traditional buses and rail trains. To address these persistent issues, existing research has proposed IMT-based solutions to reorganize onboard patrons with the closest destinations, aimed at minimizing out-of-vehicle transfers and reducing stops for boarding and alighting \citep[e.g.,][]{khan2024seamless, romea2021analysis}. 
Nevertheless, IMT-based approaches present several limitations that remain unresolved and are unlikely to be addressed in the near future:
\begin{itemize}
	\item The core technology underlying IMT—in-motion coupling—demands extremely high precision in vehicle control and necessitates complex vehicle designs, including internal passageways and automatic docking doors. These technologies are expected to be under development in the foreseeable future. 
	\item IMT is unsuitable for patrons with special needs, such as the elderly and disabled, who constitute an inextricable segment of public transit users.
	\item IMT may be inconvenient or even hazardous, particularly when patrons move between crowded vehicles traveling at high speeds in complex urban environments. 
	\item The last-mile problem persists, as most existing MAV concepts still rely on traditional access modes (e.g., walking), and patrons still have to wait outside the vehicle at origin stops.
\end{itemize}

To address these research gaps, we propose the FBT using modular mini-vehicles that reduce technological barriers by eliminating the need for in-motion coupling and transfers, as well as inter-vehicle passageway designs. To align with current and near-term technological capabilities, we conservatively restrict FBT operations to (i) \textit{stationary coupling} and (ii) \textit{in-motion decoupling of the tail trailer}, of which the technical feasibility is provided in the next section. This approach represents a cautious initial step toward implementation-ready mobility solutions. As in-motion coupling technology matures and becomes viable for public road deployment, it may offer additional benefits for specific user groups under appropriate conditions. 

The formal development of the FBT concept is presented as follows.

\section{Concept development} \label{sec_concept}
\subsection{ Fly-by transit concept}
Unlike existing modular mobility concepts, FBT contains two complementary types of autonomous mini-EV modules, as shown in Fig. \ref{fig_concept}: (i) \textit{trailers} – small-battery electric personalized pods seating 1–2 patrons (and customizable for larger capacities for special purposes); and (ii) \textit{leaders} – guided/propulsive pods equipped with high-capacity batteries and motors. Multiple trailers can run independently for short-range, individual transport, and couple physically with a leader to form a platoon for long-range and high-speed transport, while being charged en route.
\begin{figure}[!ht]
	\centering
	\includegraphics[width=0.6\linewidth]{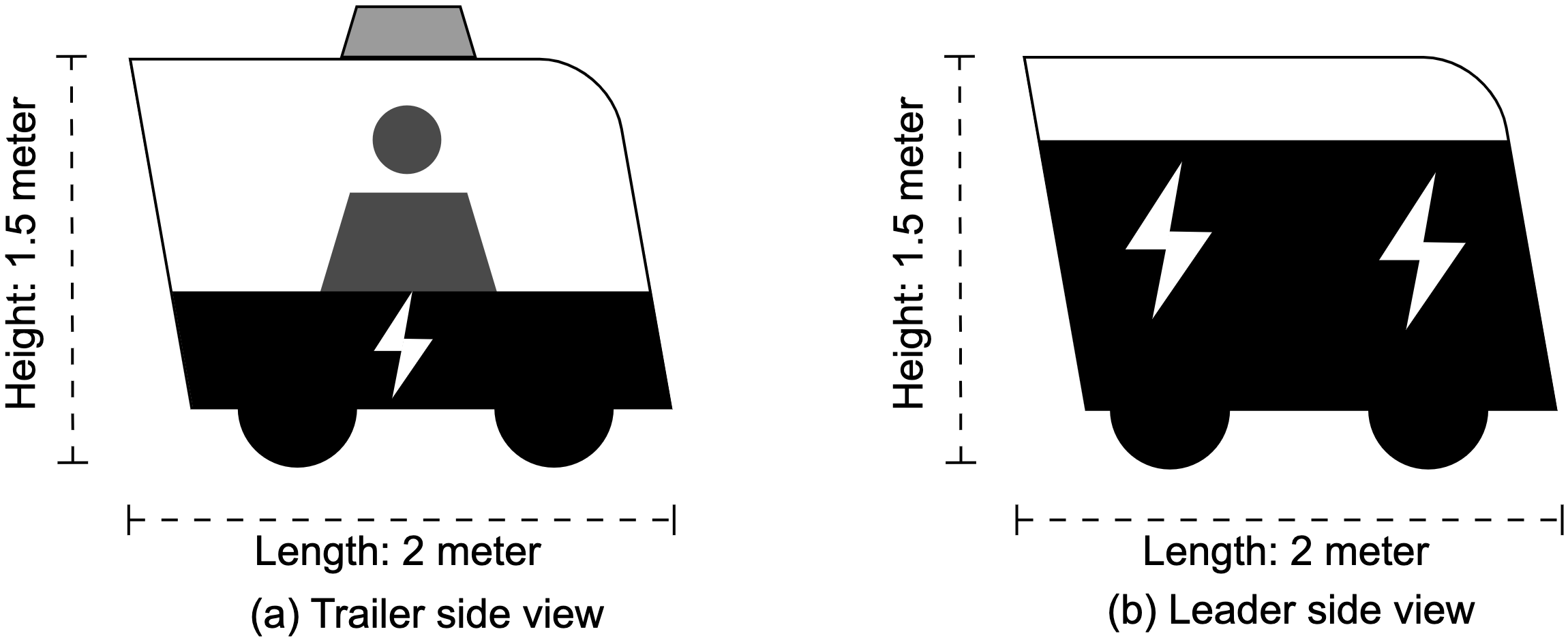}
	\caption{FBT vehicle modules: (a) trailer, (b) leader, compact design for one or two seated with no inter-vehicle passageway.}
	\label{fig_concept}
\end{figure}

FBT enables door-to-door services with minimal intermediate stops through the following operations: 
\begin{enumerate}[label=(\roman*)]
	\item Trailers operate like ride-hailing taxis, transporting individual patrons to the closest FBT stations with no detours; 
	\item At stations, trailers form into physically connected platoons \textit{in descending order of trip length} in the operational direction. This can be achieved through \textit{stationary coupling}, utilizing current technologies as elaborated in Section \ref{subsec_tech}.
	\item Platoons, propelled by the leaders, travel at high speed along arterial corridors, and \textit{only the tail trailer decouples} upon approaching its destination without delaying the convoy. As each trailer serves one or two patrons with the same destination, no transfer between moving vehicles is required. Note that the in-motion decoupling need not be confined to stations but can occur at any location closest to the patron’s destination.
	\item The leader, after all trailers decoupled, stops at the closest station for charging and awaits the next assembling and dispatching. 
\end{enumerate}

To clarify the FBT concept, Table \ref{tab_bigtable} distinguishes it from existing MAV concepts\footnote{Traditional non-modular vehicles have been utilized for a point-to-point transit system featuring both stop-less trunk and local feeder networks \citep{cortes_design_2002}. Nevertheless, their system requires passengers to walk and transfer at both ends of their trips and involves multiple intermediate stops for pick-ups and drop-offs in the local network.} by comparing vehicular characteristics, operational features, and patron experience. 

The concept most similar to FBT is EU-MSC, which is based on a \textit{single} mini-EV module, analogous to the FBT trailer. However, without collective trunk transport, the EU-MSC system would be limited to short-range individual transport services, essentially operating as a one-way carsharing system. 

Another related mobility solution operates the Dubai-MAB in two modes: trailers serve as feeders to main modules, and they couple at stations to form platoons along main lines \citep{zhang2020modular, zou2024operational}. Nevertheless, the use of a single-type MAB with a capacity of over four passengers presents several limitations relative to FBT: (a) serving multiple passengers requires detours for pick-up and drop-off, causing delays; (b) relaying on IMT to redistribute passengers to the correct trailer before detachment; and (c) trailer detachment is not restricted to the rear, requiring in-motion reconfiguration of intermediate trailers and generating additional traffic disturbances. 

From the patron's perspective, FBT provides an enhanced and seamless ride experience. A typical FBT trip comprises three segments: origin-station, station-trunk line, and trunk line-destination, in which her travel delays will be significantly reduced, thanks to the absence of intermediate stops, compared to traditional buses. Although her travel distances and times may exceed those of direct e-hailing taxis, the scalability of low-cost trailers combined with platoon-level EoS is expected to yield lower operational costs and thus more affordable fares. 

Moreover, from the perspective of external stakeholders such as general traffic participants, FBT {consumes less road space due to zero-gap vehicle platoons (than virtually coupled platoons or conventional taxis) and presents fewer disturbances and safety risks due to its operational simplicity (compared to IMT-based MAV concepts). }
The following section provides a preliminary techno-economic analysis to assess the technological and economic feasibility of FBT.

\begin{landscape}
\begin{longtable}[c]{|lllllll|}
\caption{Qualitative comparison between Fly-by transit and existing MAV concepts.}
\label{tab_bigtable}\\
\hline
\multicolumn{1}{|l|}{} & \multicolumn{2}{l|}{\textbf{Fly-by transit}} & \multicolumn{1}{l|}{\textbf{EU-MSC}} & \multicolumn{1}{l|}{\textbf{Dubai-MAB}} & \multicolumn{1}{l|}{\textbf{Beijing-MAB}} & \textbf{Singapore-MAB} \\ \hline
\endfirsthead
\endhead
\multicolumn{7}{|l|}{\textit{Vehicular characteristics}} \\ 
\hline
\multicolumn{1}{|l|}{\multirow{2}{*}{Vehicle}} & \multicolumn{1}{c|}{\includegraphics[width=0.075\linewidth]{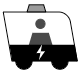}} & \multicolumn{1}{c|}{\includegraphics[width=0.08\linewidth]{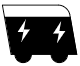}} & \multicolumn{1}{c|}{\includegraphics[width=0.2\linewidth]{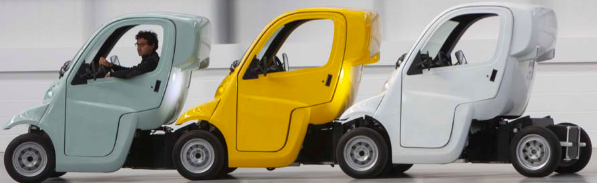}} & \multicolumn{1}{c|}{\includegraphics[width=0.1\linewidth]{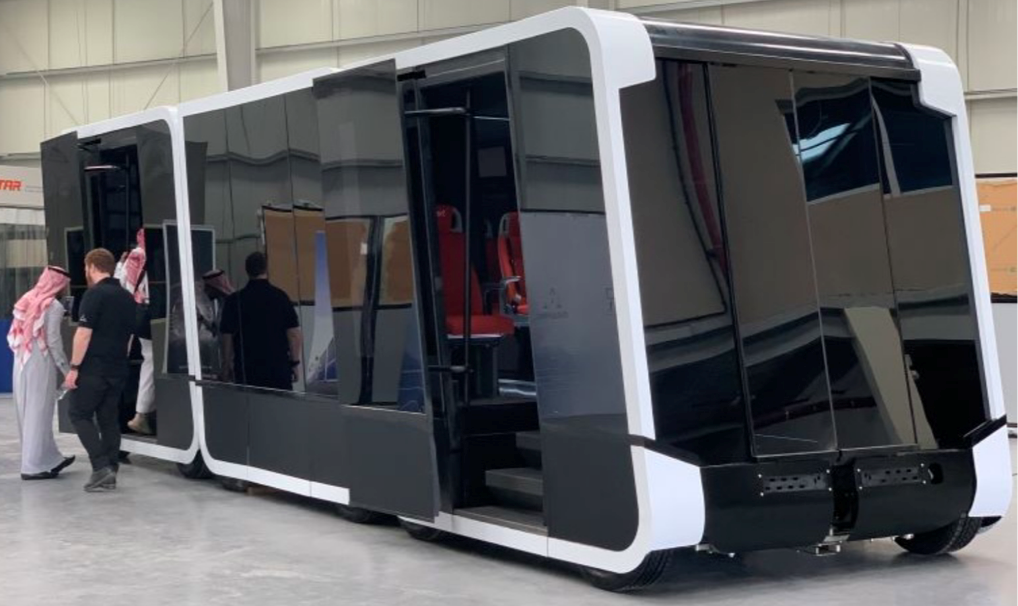}} & \multicolumn{1}{c|}{\includegraphics[width=0.1\linewidth]{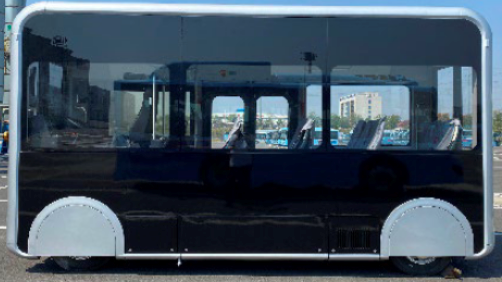}} & \multicolumn{1}{c|}{\includegraphics[width=0.1\linewidth]{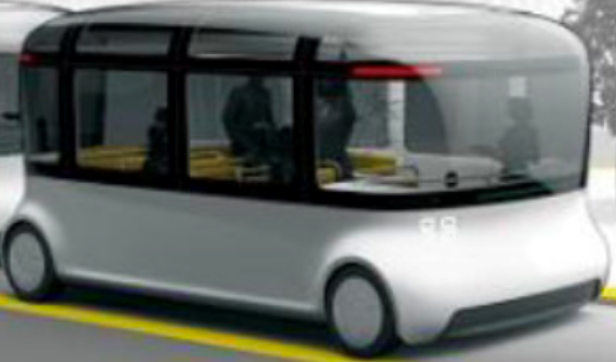}} \\ 
\cline{2-7} 
\multicolumn{1}{|l|}{} & \multicolumn{2}{l|}{\makecell[l]{- Two mini-EV modules: \\Trailers (with smaller battery,\\ low-end motor), and leaders \\(as guided vehicle);\\ - Capacity of 1-2 seated \\patrons;\\ - Compat design for seated;\\ - Battery capacity: \\$\sim$50 km range for trailers; \\$\sim$300 km range for leaders}} & \multicolumn{1}{l|}{\makecell[l]{- Single mini-EV module;\\ - Capacity of 1-3 seated \\patrons;\\ - Compat design for seated;\\ - Battery capacity:\\ over 50 km range;\\ - Prototype vehicles \\ \citep{iacobucci2022multi}}} & \multicolumn{1}{l|}{\makecell[l]{- Single electric minibus \\module;\\ - Capacity of 4-10 seated \\ patrons \& standees;\\ - Designed for moving \\standees;\\ - Inter-vehicle passage doors;\\ - Battery capacity: \\$\sim$300 km range;\\ - Vehicle model for reference:\\ Prototype vehicles}} & \multicolumn{2}{l|}{\makecell[l]{- Single electric minibus module;\\ - Capacity of 10-25 seated patrons and standees;\\ - Designed for moving standees;\\ - Battery capacity: $\sim$300 km range;\\ - Prototype vehicles \\ \citep{lin2025big,ye2024exploring};}}  \\ 
\hline
\multicolumn{1}{|l|}{Platoon} & \multicolumn{2}{l|}{\makecell[l]{- Mechanically \& \\electrically connected;\\ - Stationary coupling;\\ - In-motion decoupling \\of the tail MAV;\\ - MAVs coupled in order of \\their trip lengths; \\- No inter-vehicle passage;\\ - In-motion charging;\\ - Maximum 10 modules;}} & \multicolumn{1}{l|}{\makecell[l]{- Mechanically \& \\electrically connected;\\ - Stationary coupling;\\ - Stationary decoupling;\\ - Vacant vehicles coupled \\with one occupied for \\rebalancing only;\\ - No inter-vehicle passage;\\ - Inter-vehicle charging;\\ - Maximum 8 modules;}} & \multicolumn{2}{l|}{\makecell[l]{- Mechanically \& electrically connected;\\ - In-motion coupling \& decoupling \\ of any MAV in the platoon;\\ - MAVs coupled in no particular order;\\ - Inter-vehicle docking for passage;\\ - In-motion charging;\\ - Maximum 4-6 modules;}} & \makecell[l]{- Virtually connected;\\ - In-motion coupling \\\& decoupling of any \\MAV in the platoon;\\ - MAVs coupled in no \\particular order;\\ - No inter-vehicle passage;\\ - No in-motion charging;\\ - Maximum 10 modules;} \\ 
\hline
\multicolumn{1}{|l|}{Auto-driving} & \multicolumn{3}{l|}{\makecell[l]{Compatible for\\ - fully auto-driving;\\ - semi auto-driving with human driver on the lead vehicle;\\ - human driving;}} & \multicolumn{2}{l|}{\makecell[l]{Require\\ - fully auto-driving;\\ - high-precision automatic docking of in-motion \\vehicles;}} & \makecell[l]{Compatible for\\ - fully auto-driving;\\ - semi auto-driving with \\human driver on the lead \\vehicle;\\ - human driving;} \\ 
\hline
\multicolumn{7}{|l|}{\textit{Operational characteristics}} \\ 
\hline
\multicolumn{1}{|l|}{\makecell[l]{Operational \\schemes}} & \multicolumn{2}{l|}{\makecell[l]{- Door-to-door services;\\ - Trunk lines for stopping-less \\platoons;\\ - Ride-hailing for individuals\\ in feeder zones;\\ - Stations for coupling \\vehicles;\\ - Transfer at stations without\\ out-of-vehicle walking \\\& waiting;}} & \multicolumn{1}{l|}{\makecell[l]{- Carsharing services in forms \\of station-based or free-floating;}} & \multicolumn{2}{l|}{\makecell[l]{- Stop-to-stop services;\\ - Transit lines for stopping-less platoons;\\ - Stops for patrons walking to/from decoupled \\vehicles;\\ - In-vehicle transfers in moving platoons;}} & \makecell[l]{- Stop-to-stop services;\\ - Transit lines for platoons \\stopping at every stop;\\ - Stops for patrons walking \\to/from stopped \\platoons/vehicles;\\ - Out-of-vehicle transfers \\at stops;} \\ 
\hline
\multicolumn{7}{|l|}{\textit{Passenger experience}} \\ 
\hline
\multicolumn{1}{|l|}{\makecell[l]{Trip \\ components}} & \multicolumn{2}{l|}{\makecell[l]{- No walking;\\ - Waiting at home for pick-ups;\\ - Waiting in vehicles at origin \\ stations;\\ - Not driving the vehicle;\\ - No stopping delays \\along lines;\\ - Waiting in vehicles at transfer \\stations}} & \multicolumn{1}{l|}{\makecell[l]{- Walking to/from shared \\vehicles;\\ - Driving the vehicle by \\themselves;\\ - No en-route stopping delays;}} & \multicolumn{2}{l|}{\makecell[l]{- Walking to/from stops;\\ - Waiting out of vehicle at origin stations;\\ - Not driving the vehicle;\\ - No stopping delays along lines;\\ - Inter-vehicle transfers without waiting;}} & \makecell[l]{- Walking to/from stops;\\ - Waiting out of vehicles at \\origin stations;\\ - Not driving the vehicle;\\ - Stopping delays for \\boarding and alighting \\patrons along lines;\\ - Waiting out of vehicles \\at transfer stations} \\ 
\hline
\end{longtable}
\end{landscape}

\subsection{Preliminary techno-economic analysis (TEA)} \label{subsec_tech}
FBT’s technological feasibility rests on four near-term advancements: 
\begin{itemize}
	\item Stationary coupling, using a controlled static environment to reliably form platoons via magnetic-mechanical coupling—a practical choice given the current immaturity of in-motion coupling \citep{he2025lidar}. Prototype EU-MSCs have demonstrated the feasibility of a self-aligning, automatically locking coupling mechanism, as shown in Fig. \ref{fig_coupling}, following the ``Scharfenberg'' coupling principle \citep{boldrini2017stackable}. 
	\item In-motion decoupling, enabled by advanced V2V communication and control systems. Similar operations, called ``slipping,'' have been implemented in rail trains for over a century \citep{nold2024train}. Recent advancements in platoon splitting \citep{li2023trajectory} also support this. The proposed future research in Section \ref{sec_agenda} will further enhance this technology in the urban traffic environment. 
	\item Operation of physically connected platoons, demonstrated in the simulation and field tests of EU-MSC  \citep{gobillot2018esprit} and Dubai-MAB\footnote{\url{https://www.youtube.com/watch?v=kJlQaCIUHTI}}, which established the feasibility of inter-vehicle communication and propulsion systems, as well as the stability in various maneuvering scenarios.
	\item Autonomous or semi-autonomous driving (with in-vehicle/remote human operators), supported by sensor fusion with sub-decimeter accuracy \citep{tu2024advancements,zhang2024multisensor}, real-time data processing, and artificial intelligence (AI), evident in growing robotaxi fleets in the U.S. and China. This paper will also propose transition plans for FBT based on traditional human-driven vehicles (HVs); see Section \ref{sec_agenda}.
\end{itemize}
\begin{figure}[!ht]
	\centering
	\includegraphics[width=0.75\linewidth]{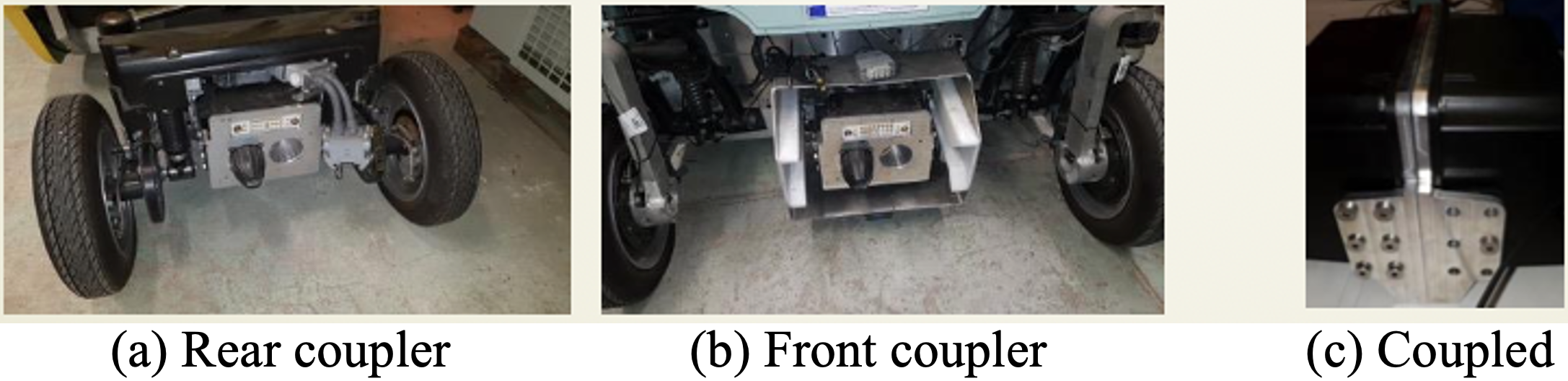}
	\caption{Coupling and decoupling system in EU-MSC \citep{gobillot2018esprit}.}
	\label{fig_coupling}
\end{figure}
\begin{figure}[!ht]
	\centering
	\includegraphics[width=0.8\linewidth]{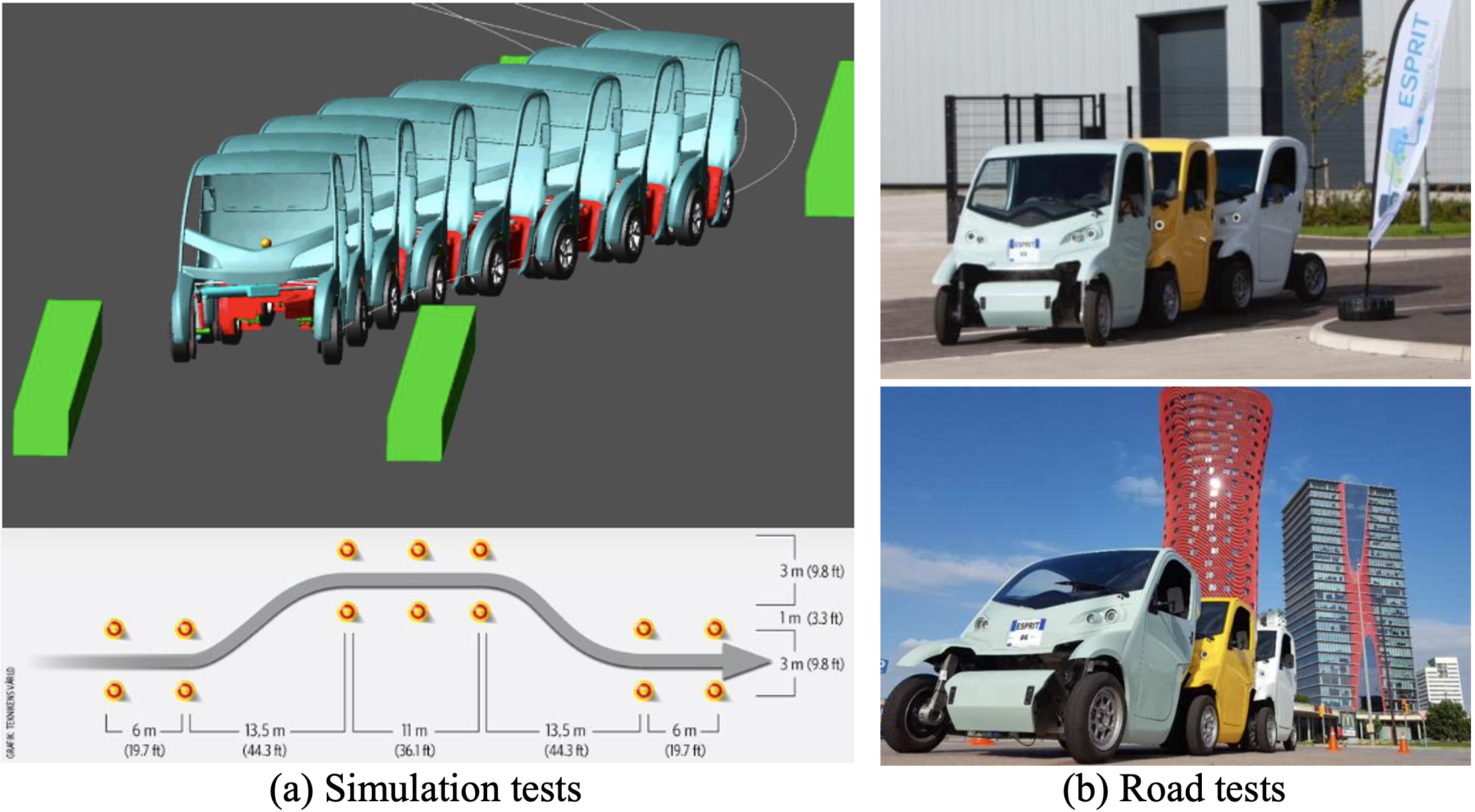}
	\caption{Stability tests of operating platoons (at up to 30 km/h) \citep{gobillot2018esprit}.}
	\label{fig:placeholder}
\end{figure}

The economic feasibility of FBT is primarily derived from the low cost of trailers and the platooning operations: 
\begin{itemize}
	\item Trailers are designed with small batteries (e.g., $\leq$50 km range) and low-power motors to function specially as short-range feeders on local roads, rather than as all-purpose vehicles across entire networks, which typically require battery capacities exceeding 300 km and freeway-capable motors. This minimalist design significantly lowers capital costs and enhances the scalability of FBT system. Preliminary estimates in Appendix \ref{appen_trailer_cost} indicate that the cost of six FBT trailers is approximately equivalent to that of a single conventional electric taxi or four commercially available mini-EVs. This cost advantage enables the deployment of six times as many trailers as ride-hailing taxis, thereby decreasing the pickup time.
	\item The leader-led platoons of physically coupled trailers achieve Economies of Scale (EoS) through shared use of the leader's large battery and high-performance motor, which enables in-motion charging and high-speed line-haul transportation. Greater EoS arises in FBT systems as demand increases: (1) shorter headways, which decrease average waiting times; (2) closer station spacings, which lower average access times; and (3) improved network connectivity, which facilitates easier transfers. 
	The following sections further present theoretical justification and numerical evidence for the EoS achieved by FBT.
\end{itemize}

The above TEA indicates that FBT combines technological feasibility and cost efficiency to enable scalable and sustainable deployment. 
Building on these promising findings, the next section presents a rigorous, quantitative evaluation of FBT performance.

\section{Proof of concept} \label{sec_proof_of_concept}
For proof of concept, this section develops a corridor model to evaluate the performance of optimized FBT systems. Corridors are fundamental building blocks of urban transit networks and play a central role in large cities such as Hong Kong and Beijing. Specifically, we consider a ring-shaped corridor\footnote{Ring-shaped corridors are common in cities with ring–radial network structures (e.g., Chengdu). Other corridor configurations can be represented using the ring-corridor model by treating each operational direction as one half of the ring.} with length $L$ km and width $W$ km, as shown in Fig. \ref{fig_ring}, subject to a uniform demand density of $\lambda$ trips/km$^2$. Patron trip lengths are assumed to follow a uniform distribution, $\ell \sim \mathrm{U}[a, b]$, where $0 \leq a \leq b \leq L/2$. 

This stylized setting is designed to generate general insights to overarching questions, such as whether and to what extent FBT exhibits EoS, how FBT compares to alternative systems, and the magnitude of its potential benefits, rather than to guide specific applications. For practical scenarios, models incorporating realistic (non-uniform) demand distributions and heterogeneous corridor and network designs are under development in our ongoing work, which is beyond the scope of this paper.
\begin{figure}[!ht]
	\centering
	\includegraphics[width=0.9\linewidth]{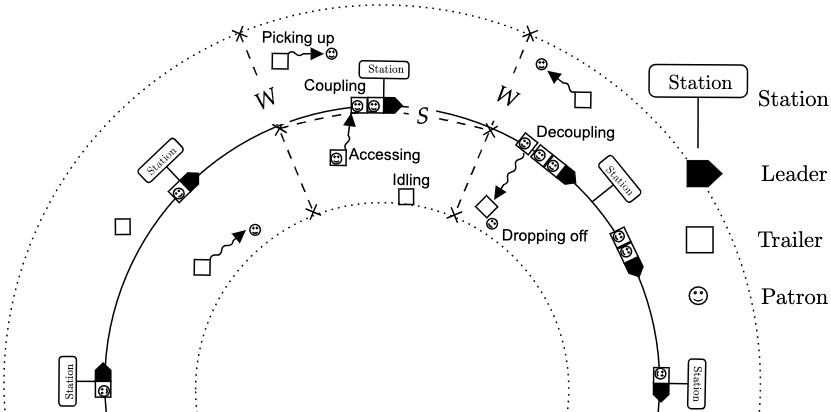}
	\caption{Fly-by transit in a ring corridor (depicting half).}
	\label{fig_ring}
\end{figure}

The corridor design problem involves determining several key variables that characterize FBT systems. These include station spacing $S$ (km/station), service headway $H$ (hours/dispatch), platoon size $J$ (trailers per platoon), trailer fleet size $F_{\mathrm{t}}$ (i.e., the number of trailers), and leader fleet size $F_{\mathrm{l}}$ (i.e., the number of leaders). Notably, some of these decision variables are interdependent, as demonstrated in the subsequent analysis. Consequently, the number of independent decision variables is reduced.

Table \ref{tab_notations} summarizes the notations used in this section.
\begin{table}[!ht]
	\centering
	\caption{Notations}
	\label{tab_notations}
	\resizebox{\columnwidth}{!}{%
		\begin{tabular}{lll|lll}
			\hline
			Variable &
			Description &
			\begin{tabular}[c]{@{}l@{}}Baseline value \\ (unit)\end{tabular} &
			Variable &
			Description &
			\begin{tabular}[c]{@{}l@{}}Baseline value\\ (unit)\end{tabular} \\ \hline
			$a, b$ &
			Demand parameters &
			$0, L/2$ (km) &
			\begin{tabular}[c]{@{}l@{}}$t_c, t_e, t_p$\\ $t_a, t_w$\end{tabular} &
			\begin{tabular}[c]{@{}l@{}}Average time trailers spent \\ in each state\end{tabular} &
			/ (hr/trailer) \\
			$f_i$ &
			Density of idle trailers &
			/ (per km$^2$) &
			$U_o$ &
			Operator cost &
			/ (\$/trip) \\
			$F_{\mathrm{t}}, F_{\mathrm{l}}$ &
			Fleet sizes of trailers, leaders &
			/ &
			$U_p$ &
			Patron cost &
			/ (hr/trip) \\
			$H$ &
			Headway &
			/ (hr) &
			$v, V$ &
			Speed of individual trailers, leader &
			20, 50 (km/hr) \\
			$J, J_{\mathrm{max}}$ &
			Platoon size, its maximum value &
			/, 10 (trailers) &
			$Z$ &
			System cost &
			/ (hr/trip) \\
			$k$ &
			Network-related parameter &
			0.63 &
			$\alpha$ &
			Patrons' value of time &
			\{5, 25\} (\$/hr) \\
			$K_{\mathrm{t}}, K_{\mathrm{l}} $ &
			Vehicle distance traveled &
			/ (veh-km/hr) &
			$\lambda$ &
			Demand density &
			1--300 (trips/km$^2$/hr) \\
			$\ell_i$ &
			Trip length of trailer $i$ &
			$[a, b]$(km) &
			$\pi_{\mathrm{tf}}, \pi_{\mathrm{lf}}$ &
			Unit capital cost of trailers, leaders &
			\begin{tabular}[c]{@{}l@{}}0.083 (\$/trailer/hr)\\ 10.4 (\$/leader/hr)\end{tabular} \\
			$L, W$ &
			Corridor length, width &
			20, 5 (km) &
			$\pi_{\mathrm{tk}}, \pi_{\mathrm{lk}}$ &
			Unit operational cost of trailers, leaders &
			\begin{tabular}[c]{@{}l@{}}0.038 (\$/trailer/km)\\ 0.3 (\$/leader/km)\end{tabular} \\
			\begin{tabular}[c]{@{}l@{}}$n_c, n_e, n_i $\\ $n_p, n_a, n_w$\end{tabular} &
			Number of trailers in six states &
			/ (trailers) &
			$\pi_S$ &
			Unit operation cost of stations &
			7.7 (\$/station/hr) \\
			$S$ &
			Station spacing &
			/ (km) &
			\begin{tabular}[c]{@{}l@{}}$\mu_c, \mu_e, \mu_i$ \\ $\mu_p, \mu_a, \mu_w$\end{tabular} &
			Change rates of trailers in five states &
			/ (trailers/hr) \\ \hline
		\end{tabular}%
	}
\end{table}

\subsection{Steady-state equations}
The optimal design of FBT concerns systems in steady states, as illustrated in Fig. \ref{fig_steady_state}. In this context, trailers transition through six distinct operational states: 
\begin{enumerate}[label=(\roman*)]
	\item Cruising (denoted by subscript $c$): The trailer travels as part of a platoon until it decouples for egress.
	\item Egressing ($e$): As the trailer approaches its destination, it decouples from the main convoy and delivers the patron to their destination.
	\item Idling ($i$): After the patron alights, the trailer enters an idle state while awaiting a new assignment.
	\item Picking up ($p$): Upon receiving a matched request, the trailer proceeds to pick up the patron.
	\item Accessing ($a$): After the patron boards, the trailer heads to the nearest station.  
	\item Waiting ($w$): The trailer, with the patron on board, arrives at the station and waits for coupling and dispatch in a platoon. 
\end{enumerate} 
\begin{figure}[!ht]
	\centering
	\includegraphics[width=0.75\linewidth]{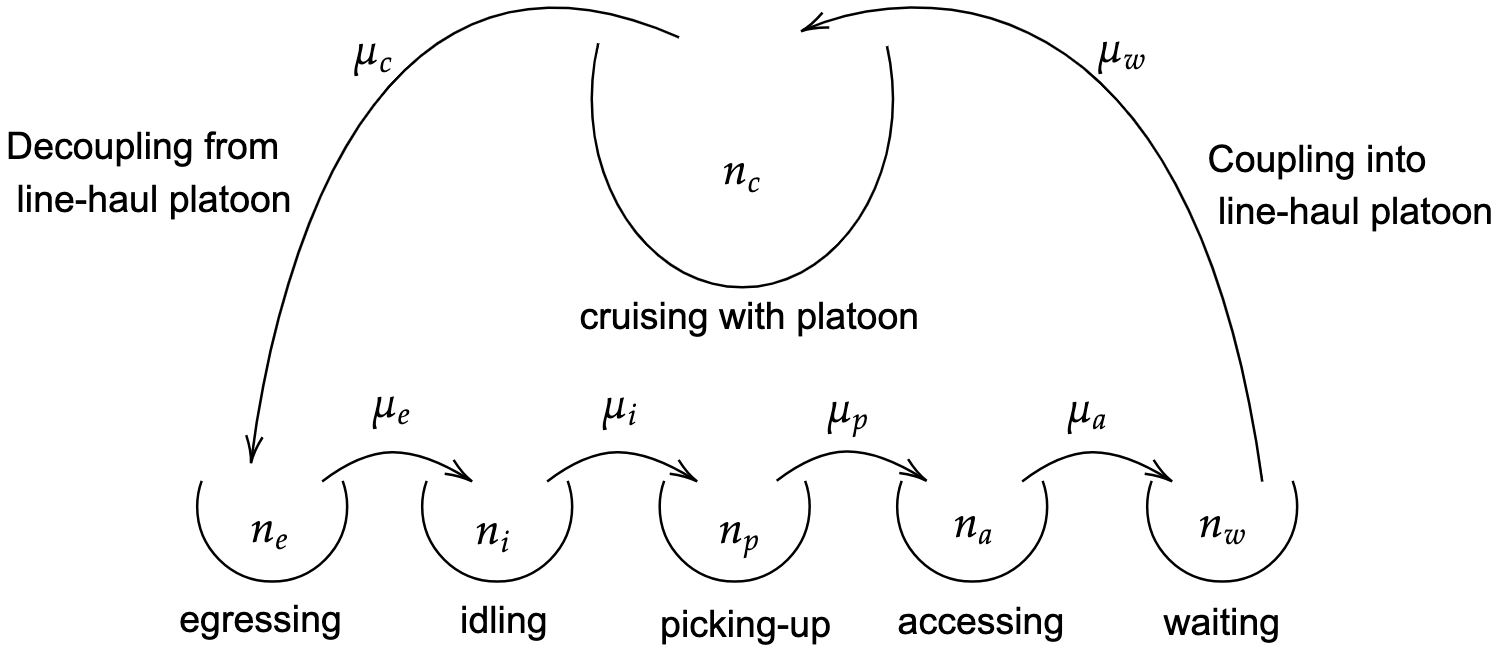}
	\caption{Transition of trailers in steady states.}
	\label{fig_steady_state}
\end{figure}

These states collectively characterize the operational cycle of trailers within the FBT system. Since the system is isotropic, our analysis unit can be the coverage of each station, i.e., $ S\times W$, in which the number of trailers in each state is thus represented by $n_{c}, n_e, n_i, n_p, n_a$, and $n_w$. These trailers change their states at rates of $\mu_c, \mu_e, \mu_i, \mu_p, \mu_a$, and $\mu_w$. 

When the system operates at steady states without oversaturation, as shown in Fig. \ref{fig_steady_state}, we have the following equations :
\begin{align}
	& \mu_c = \mu_e = \mu_i = \mu_p = \mu_a = \mu_w , \label{eq_steady} \\
	& \mu_i = \lambda SW, \label{eq_begining} \\
	& \mu_e = \lambda SW, \label{eq_ending} \\
	& \mu_w = \frac{J}{H}, \label{eq_dispatching}
\end{align}
where equation (\ref{eq_steady}) indicates steady-state conditions, equations (\ref{eq_begining}, \ref{eq_ending}) connect trailers' transitions with demand generating and accomplishing processes; and equation (\ref{eq_dispatching}) means the dispatching flow per unit time.

According to Little's law in queueing theory, the following relationships hold:
\begin{align} \label{eq_little_law}
	& \mu_c = \frac{n_c}{t_c}, &
	& \mu_e = \frac{n_e}{t_e}, &
	& \mu_p = \frac{n_p}{t_p}, &
	& \mu_a = \frac{n_a}{t_a}, &
	& \mu_w = \frac{n_w}{t_w}, &
\end{align}
where $t_c, t_e, t_p, t_a$, and $t_w$ are the average time trailers spent in respective states, and they can be estimated as follows. 

The $t_c$ is the average travel time with the platoon, and can therefore be calculated by dividing the average trip distance $\frac{a+b}{2}$ by the cruising speed $V$:
\begin{align} \label{eq_tc}
	t_c = \frac{a+b}{2V}.
\end{align}

Since the trailer can decouple from the convoy at a point that is closest to her destination, the average travel distance of a decoupling trailer is $W/4$, leading to 
\begin{align}
	t_e = \frac{W}{4v}, \label{eq_te}
\end{align}
where $v$ is the average speed of trailers on local roads. 

The $t_p$ is the average travel time of a matched trailer to her pick-up request, which, according to \cite{daganzo2019public}, can be gauged by 
\begin{align}
	t_p = \frac{k}{v}\sqrt{\frac{SW}{n_i}} \label{eq_tp}
\end{align}
where $k$ is a dimensionless value reflecting network structure (e.g., $k=0.63$ in grid networks). The term $k\sqrt{\frac{SW}{n_i}}$ represents the expected distance of a random request to the closest idle vehicle in a pool of $n_i$ trailers.

The $t_a$ is the average travel time of a trailer traveling from the pick-up location to the closest station, which can be estimated by 
\begin{align}
	t_a = \frac{S+W}{4v}. \label{eq_ta}
\end{align}

The $t_w$ is the average waiting time at the station to be coupled into platoons, which is being dispatched every $H$ minutes; thus, we have
\begin{align}
	t_w = H/2. \label{eq_tw}
\end{align}

\subsection{Operator cost}
The operator cost, $U_o$, can be expressed by the average cost per trip as follows
\begin{align}
	U_o = & \frac{1}{\lambda LW} \left[ \pi_S \frac{L}{S} + \left(\pi_{\mathrm{tf}} F_{\mathrm{t}} + \pi_{\mathrm{tk}} K_{\mathrm{t}}\right) + \left(\pi_{\mathrm{lf}} F_{\mathrm{l}} + \pi_{\mathrm{lk}} K_{\mathrm{l}} \right) \right], \label{eq_Uo}
\end{align}
where $\pi_S$ denotes the amortized unit cost per operating hour for station construction and maintenance, and the term $\frac{L}{S}$ gives the total number of stations along the corridor. The second and third terms in brackets represent the operational costs of trailers and leaders, each comprising two components: capital cost and energy consumption cost. 

Specifically, $\pi_{\mathrm{tf}}$ and $\pi_{\mathrm{tk}}$ are the unit costs for trailers, corresponding to the amortized purchase cost of the fleet $F_{\mathrm{t}}$ and energy consumption over the vehicle distance $K_{\mathrm{t}}$. Similarly, $\pi_{\mathrm{lf}}$ and $\pi_{\mathrm{lk}}$ are the unit costs for leaders, reflecting the amortized purchase cost of the fleet $F_{\mathrm{l}}$ and energy consumption over the vehicle distance $K_{\mathrm{l}}$. The formulations for $F_{\mathrm{t}}$, $K_{\mathrm{t}}$, $F_{\mathrm{l}}$, and $K_{\mathrm{l}}$ are provided below.

\subsubsection{Trailer fleet $F_{\mathrm{t}}$ and vehicle distance $K_{\mathrm{t}}$}
Substituting $t_c, t_e, t_p, t_a, t_w$ into (\ref{eq_little_law}) and combining with (\ref{eq_steady}--\ref{eq_dispatching}) yields
\begin{align}
	& n_c = \lambda SW  \frac{a+b}{2V}, \\
	& n_e =  \lambda S W \frac{W}{4v}, \\
	& n_p =  \lambda S W \frac{k}{v}\sqrt{\frac{SW}{n_i}}, \\
	& n_a = \lambda S W \frac{S+W}{4v}, \\
	& n_w = \lambda S W \frac{H}{2} = \frac{J}{2}. 
\end{align}

If we define the density of idle trailers as $ f_i = \frac{n_i}{SW}$ and replace $H$ with $\frac{J}{\lambda SW}$, we can now express the total fleet size of individual trailers by
\begin{align} \label{eq_Ft}
	F_{\mathrm{t}}(S, J, f_i) = &\frac{L}{S} \left( n_c + n_i + n_e + n_p + n_a + n_w \right) \notag \\
	= & \lambda LW \left ( \frac{a+b}{2V} + \frac{f_i}{\lambda} + \frac{k}{v}\sqrt{\frac{1}{f_i}} + \frac{S+2W}{4v} + \frac{J}{2\lambda S W} \right ).
\end{align}

We next formulate $K_{\mathrm{t}}$ as Eq. (\ref{eq_Kt}) by considering only the vehicle distance traveled by individually operating trailers, as they are responsible for their own energy consumption. The distance traveled while coupled in platoons is excluded, since propulsion in these segments is provided by the leaders. 
\begin{align} \label{eq_Kt}
	K_{\mathrm{t}} (S,f_i) = \lambda LW \left( k\sqrt{\frac{1}{f_i}} + \frac{S+2W}{4} \right),
\end{align}
where $\lambda L W $ yields the total number of trailers per operational hour, and the terms in parentheses represent the distances traveled by an average trailer, with the second term combining the travel distances in accessing and egressing states. Note that the idling distance is ignored.

\subsubsection{Leader fleet $F_{\mathrm{l}}$ and vehicle distance $K_{\mathrm{l}}$}
For the leaders, we estimate the fleet size by 
\begin{align} \label{eq_Fl}
	F_{\mathrm{l}}(S, J) = & \frac{L}{S} \left(\frac{1}{H} \frac{\mathrm{E} \left( \max_{j=1,...,J} \left\{ \ell_j\right\} \right) }{V} + 1\right) \notag \\
	= & \frac{\lambda LW}{J} \frac{\mathrm{E} \left( \max_{j=1,...,J} \left\{ \ell_j\right\} \right) }{V} + \frac{L}{S},
\end{align}
where the terms in parentheses denote the number of leaders per station. The first term in parentheses represents the number of leaders required per cycle. The average cycle time for each leader is $\frac{\mathrm{E} \left( \max_{j=1,..., J} \left\{ \ell_j\right\} \right) }{V}$, where $V$ is the traveling speed of the leader/platoon and $\mathrm{E} \left( \max_{j=1,..., J} \left\{ \ell_j\right\}\right)$ returns the expected travel distance of a dispatched leader before returning to station for another dispatch, which is the maximum trip distance of all trailers coupled with the leader in the platoon. The constant 1 in parentheses accounts for one standby leader at each station to ensure continuous availability. 

Under the assumption of uniform trip length distribution, we have $\mathrm{E} \left( \max_{j=1,..., J} \left\{ \ell_j\right\}\right) = \frac{a+bJ}{1+J}$, of which the derivation is given in Appendix \ref{appen_leader_distance}. Therefore, the leaders' fleet size can be rewritten as
\begin{align} \label{eq_Fl}
	F_{\mathrm{l}}(S, J) = \frac{ \lambda LW(a+bJ) }{VJ(1+J)} + \frac{L}{S}. 
\end{align}

Leaders' vehicle distance $K_{\mathrm{l}}$ can also be obtained as
\begin{align} \label{eq_Kl}
	K_{\mathrm{l}}(J) = \frac{L}{S} \frac{1}{H} \frac{a + bJ}{1 + J} = \frac{\lambda LW}{J} \frac{a + bJ}{1 + J}.
\end{align}

Substituting (\ref{eq_Ft}), (\ref{eq_Kt}), (\ref{eq_Fl}), and (\ref{eq_Kl}) into (\ref{eq_Uo})\footnote{Note that the leader’s energy consumption should vary over the platoon’s trip as trailers progressively decouple. For analytical tractability, we approximate this with an average unit cost per distance $\pi_{\mathrm{lk}}$. }, the operator cost $U_o$ becomes
\begin{align}
	U_o (S, J, f_i) =  & \frac{2\left(\pi_{S}+\pi_{\mathrm{lf}}\right)+\pi_{\mathrm{tf}}J}{2\lambda SW}+\pi_{\mathrm{tf}}\left(\frac{a+b}{2V}+\frac{f_{i}}{\lambda}\right)+\left(\frac{\pi_{\mathrm{tf}}}{v}+\pi_{\mathrm{tk}}\right)\left(k\sqrt{\frac{1}{f_{i}}}+\frac{S+2W}{4}\right) \notag \\
	& +\frac{\pi_{\mathrm{lk}}}{J}\frac{a+bJ}{1+J}\left(\frac{1}{V}+1\right).
\end{align}

\subsection{Patron cost}
An average patron's cost $U_p$ is composed of five components dependent on five trip segments of the entire trip: (i) at-home waiting time before being picked up, which equals $t_p$ in (\ref{eq_tp}), (ii) access time to the closest station, which is equivalent to $t_a$ in (\ref{eq_ta}), (iii) at-station in-vehicle waiting time, i.e., $t_w$ in (\ref{eq_tw}), (iv) in-platoon line-haul time, i.e., $t_c$ in (\ref{eq_tc}), and (v) egress time for decoupling from the platoon to her destination, i.e., $t_e$ in (\ref{eq_te}).

Therefore, we have 
\begin{align} \label{eq_Up}
	U_p (S,J,f_i) = \frac{k}{v} \sqrt{\frac{1}{f_i}} + \frac{S+2W}{4v} + \frac{J}{2\lambda SW} + \frac{a+b}{2V}.
\end{align}

\subsection{Optimal design problem}
With the above results, we construct the following optimal design model to minimize the generalized average system cost $Z$ as a weighted sum of operator and patron costs.
\begin{subequations} \label{eq_optimal_design}
	\begin{align} \label{eq_obj}
		\minimize_{S,J,f_i} Z(S,J,f_i) = \frac{1}{\alpha}U_o(S,J,f_i) + U_p(S,J,f_i),
	\end{align}
	subject to,
	\begin{align}
		S, f_i > 0, J \leq J_{\mathrm{max}},J \in \mathbb{N}^+,
	\end{align}
\end{subequations}
where $\alpha$ represents patrons' average value of time and is used to convert monetary costs into time units; and $J_{\mathrm{max}}$ denotes the maximum number of trailers permitted to be coupled in a single platoon, which is associated with station capacity or mobility requirements imposed by road curvature or the propulsion capacity of the leader vehicle.

To solve (\ref{eq_optimal_design}), we notice that if we treat $J$ as given, the function $Z(S,f_i|J)$ is convex concerning $S$ and $f_i$, which can be verified to see $\frac{\partial^2 Z(S,f_i|J)}{\partial f_i^2}>0$ and $\frac{\partial^2 Z(S,f_i|J)}{\partial S^2}>0$. Therefore, we can use the first-order conditions for minimizing $Z(S,f_i|J)$:
\begin{subequations}
	\begin{align}
		\frac{\partial Z(S,f_i|J)}{\partial f_i} & = \frac{\pi_{\mathrm{tf}}}{\alpha\lambda}-\left(\frac{\pi_{\mathrm{tf}}+\alpha}{v\alpha}+\frac{\pi_{\mathrm{tk}}}{\alpha}\right)\frac{k}{2}\left(f_{i}\right)^{-\frac{3}{2}} = 0, \\
		\frac{\partial Z(S,f_i|J)}{\partial S} & =\frac{\pi _\mathrm{tf}+\pi _\mathrm{tk}v+\alpha}{4v\alpha}-\frac{(\pi _\mathrm{tf}+\alpha )J+2(\pi _\mathrm{lf}+\pi_S)}{2\alpha\lambda WS^{2}}=0,
	\end{align}
\end{subequations}
to obtain the following closed-form solutions conditional on the given $J$:
\begin{subequations}
	\begin{align}
		f_{i}^{*} & =\left[\frac{k\lambda\left(\pi_{\mathrm{tf}}+\pi_{\mathrm{tk}}v+\alpha\right)}{2\pi_{\mathrm{tf}}v}\right]^{\frac{2}{3}}, \label{eq_f_i*} \\
		S^{*}(J) & =\left[\frac{2v \left(( \pi _\mathrm{tf} +\alpha)J +2(\pi _\mathrm{lf} +\pi _{S}) \right) }{\mathnormal{\lambda W}\left( \pi _\mathrm{tf} +\pi _\mathrm{tk} v+\alpha \right)}\right]^\frac{1}{2}. \label{eq_S*}
	\end{align}
\end{subequations}

With (\ref{eq_f_i*}, \ref{eq_S*}) in hands, the optimal $J^*$ can be efficiently found via a line search:
\begin{align}
	J^* = {\arg\min}_{J = 1, 2, ..., J_{\mathrm{max}}} Z(S^*(J),f^*_i,J).
\end{align}

\begin{proposition} \label{prop_EoS}
	When $J^*$ is bounded at the boundary $J_{\mathrm{max}}$, the optimized FBT systems exhibit EoS proportional to $\lambda^{-\frac{1}{2}}$. 
\end{proposition}
\begin{proof}
	When $J^*$ is bounded at the boundary $J_{\mathrm{max}}$, substituting $f_{i}^{*}$ and $S^{*}(J_{\mathrm{max}})$ into $Z$ in (\ref{eq_obj}) yields $Z^*(J_{\mathrm{max}}) = A\lambda^{-\frac{1}{2}}+B\lambda^{-\frac{1}{3}} + C$, where $A, B$, and $C$ are given non-zero constants. Therefore, the leading-order behavior of $Z^*(J_{\mathrm{max}})$ is is governed by the $A\lambda^{-\frac{1}{2}}$ term, i.e., $Z^*(J_{\mathrm{max}}) \propto (\frac{1}{\lambda})^{\frac{1}{2}}$. This result implies that the optimized FBT's average system cost decreases with increasing demand density on the order of $\lambda^{-\frac{1}{2}}$.
\end{proof}
\begin{proposition} \label{prop_Up}
	When $J^*$ is bounded at the boundary $J_{\mathrm{max}}$, FBT offers patrons the average trip time or speed approaching that of taxi service as demand increases. 
\end{proposition}
\begin{proof}
	When $J^*$ is bounded at the boundary $J_{\mathrm{max}}$, substituting $f_{i}^{*}$ and $S^{*}(J_{\mathrm{max}})$ into $U_p$ in (\ref{eq_Up}) yileds
	\begin{align*}
		U^*_{p}(J_{\mathrm{max}})= & \left(\frac{k}{v}\right)^{\frac{2}{3}}\left[\frac{2\pi_{\mathrm{tf}}}{\lambda\left(\pi_{\mathrm{tf}}+\pi_{\mathrm{tk}}v+\alpha\right)}\right]^{\frac{1}{3}}+\frac{W}{2v}+\frac{a+b}{2V} \notag \\ 
		& +\frac{1}{2\sqrt{2v}}\left[\frac{\left((\pi_{\mathrm{tf}}+\alpha)J_{\mathrm{max}}+2(\pi_{\mathrm{lf}}+\pi_{S})\right)}{\mathnormal{\lambda W}\left(\pi_{\mathrm{tf}}+\pi_{\mathrm{tk}}v+\alpha\right)}\right]^{\frac{1}{2}} \notag \\
		& +\frac{J_{\mathrm{max}}}{2}\left[\frac{\left(\pi_{\mathrm{tf}}+\pi_{\mathrm{tk}}v+\alpha\right)}{2\lambda Wv\left((\pi_{\mathrm{tf}}+\alpha)J_{\mathrm{max}}+2(\pi_{\mathrm{lf}}+\pi_{S})\right)}\right]^{\frac{1}{2}},
	\end{align*}
	where the last two terms approach zero as demand density $\lambda$ increases; and the magnitudes of the rest terms are comparable to those in optimized taxi services; see Appendix \ref{appen_taxi_bus}.
\end{proof}

\begin{remark}
	Proposition (\ref{prop_EoS}) suggests that FBT has a more substantial EoS than taxis, whose average system cost scales on the order of $\lambda^{-\frac{1}{3}}$ \citep{daganzo2019public}; see also Appendix \ref{appen_taxi_bus}. Proposition (\ref{prop_Up}) indicates that, under sufficient demand, optimized FBT systems can provide rapid door-to-door service comparable to taxis. Given FBT's larger EoS, its operational costs are expected to be significantly lower than those of taxis. 
	
	By contrast, the door-to-door speed in conventional bus systems asymptotically approaches cycling speed as demand increases, even if one assumes zero operating costs and infinitely high service frequency, as demonstrated by \cite{daganzo2019public}; see also Appendix \ref{appen_taxi_bus} for a brief proof.
	
	Consequently, FBT has the potential to bridge the urban mobility gap, being more affordable than taxis and faster than buses. 
\end{remark}
\begin{remark}
	We also notice that the optimal density of idle trailers, $f^*_i$ in (\ref{eq_f_i*}), depends exclusively on predetermined demand density and socio-economic parameters. It is independent of design variables such as platoon size $J$ and station spacing $S$. This finding suggests that $f_i$ can serve as a convenient proxy for monitoring system operations: if $f_i < f^*_i$, additional trailers should be deployed; otherwise, if $f_i > f^*_i$, the system may be oversupplied with trailers. 
\end{remark}


\subsection{Numerical examples}
Using the baseline values in Table \ref{tab_notations}, we run numerical experiments to compare an optimized FBT system with optimized bus and taxi benchmarks. The models for the benchmarks are presented in Appendix \ref{appen_taxi_bus}. 

To ensure fair comparisons, we calibrate the FBT cost parameters relative to the market prices of mini-EVs and electric buses. We assume vehicles in all systems are equipped with comparable autonomous driving technology within the planning horizon, so their cost differences primarily reflect variations in vehicle size, battery capacity, and propulsion system. 

Specifically, we set the unit capital and operational costs of trailers to two third of a typical mini-EV's market price ($\$3000$ per min-EV) (See the reasoning in Appendix \ref{appen_trailer_cost}), i.e., $\pi_{\mathrm{tf}} = 0.083 $ (\$/trailer/hr) (calculated by $\frac{\$2000 \text{/trailer}}{8 \text{ year} \times 300 \text{ days} \times 10 \text{ hrs}}$) and $\pi_{\mathrm{tk}} = 0.038$  (\$/trailer/km). For leaders, the unit capital and operational costs are set to match those of a standard 12 m electric bus\footnote{\url{https://mobilityforesights.com/product/hong-kong-electric-bus-market}}, resulting in $\pi_{\mathrm{lf}} = 10.4 $ (\$/leader/hr) (computed by $\frac{\$250000 \text{/leader}}{8 \text{ year} \times 300 \text{ days} \times 10 \text{ hrs}}$) and $\pi_{\mathrm{lk}} = 0.3$ (\$/leader/km). The unit cost of the FBT station $\pi_{S}=\$7.7$ per station per operation hour is set based on that of the Bus Rapid Transit (BRT) station.

The following analyses are conducted in high-wage ($\alpha = \$25 $/hr) and low-wage ($\alpha = \$5 $/hr) cities with demand densities varying in the range of $ \lambda \in [1, 300]$ trips/km$^2$/hr. Parametric analysis will be conducted to examine different unit cost settings.

\subsubsection{Economies of scale}
Our numerical results confirm that the FBT system achieves an EoS substantially greater than that of taxis and even surpasses that of bus services. Fig. \ref{fig_EoS} illustrates the percentage change in average system cost per trip across varying demand levels, using $\lambda = 1 $ (trips/km$^2$/hr) as the reference. The average system cost per trip for FBT decreases remarkably as demand increases in both high and low-wage cities. In contrast, taxis exhibit only modest EoS, while bus systems exhibit an EoS that falls between those of taxi and FBT systems throughout the demand spectrum. 

This superior EoS in the FBT corridor arises from two key factors: (i) reduced service headway with increasing demand, which shortens average waiting times at stations; and (ii) decreased station spacing as demand rises, resulting in shorter average access distances and times; see Fig. \ref{fig_EoS_explanation}(a) and (b). Although bus transit benefits from these same mechanisms, FBT is distinct in that reduced station spacing does not impede platoon operations or increase in-vehicle time for onboard patrons, thanks to the absence of intermediate stops. In conventional bus systems, however, closer stop spacing lowers walking distances but increases stopping delays for vehicles and onboard patrons, thereby compromising the overall EoS. Note that the EoS in taxi services arises solely from the shared or cycled use of vehicles, a property that is also preserved in FBT systems.
\begin{figure}[!htb]
	\centering
	\subfigure[High-wage cities $\alpha = \$25$/hr]{\includegraphics[width=75mm]{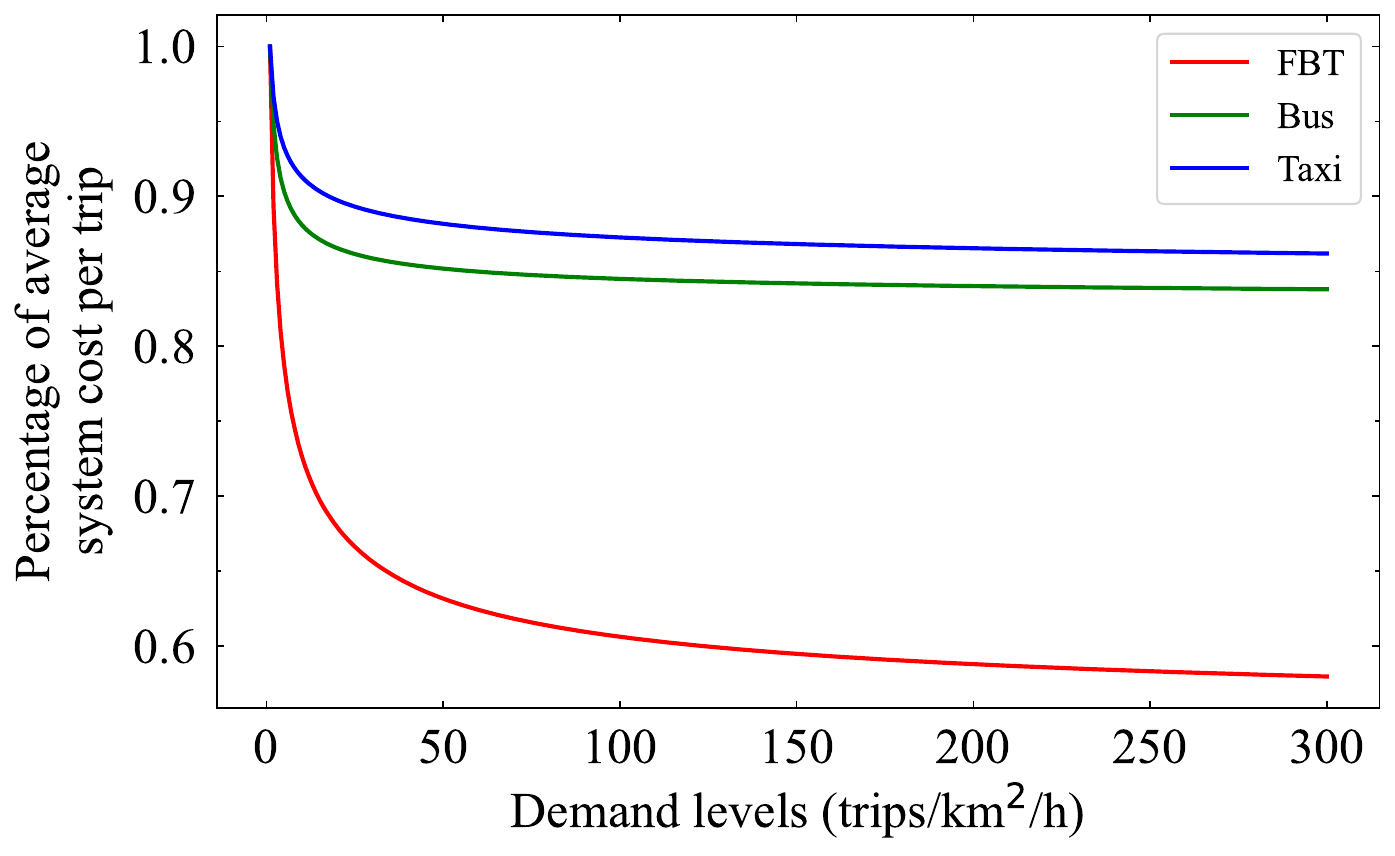}}
	\hfill
	\subfigure[Low-wage cities $\alpha = \$5$/hr]{\includegraphics[width=75mm]{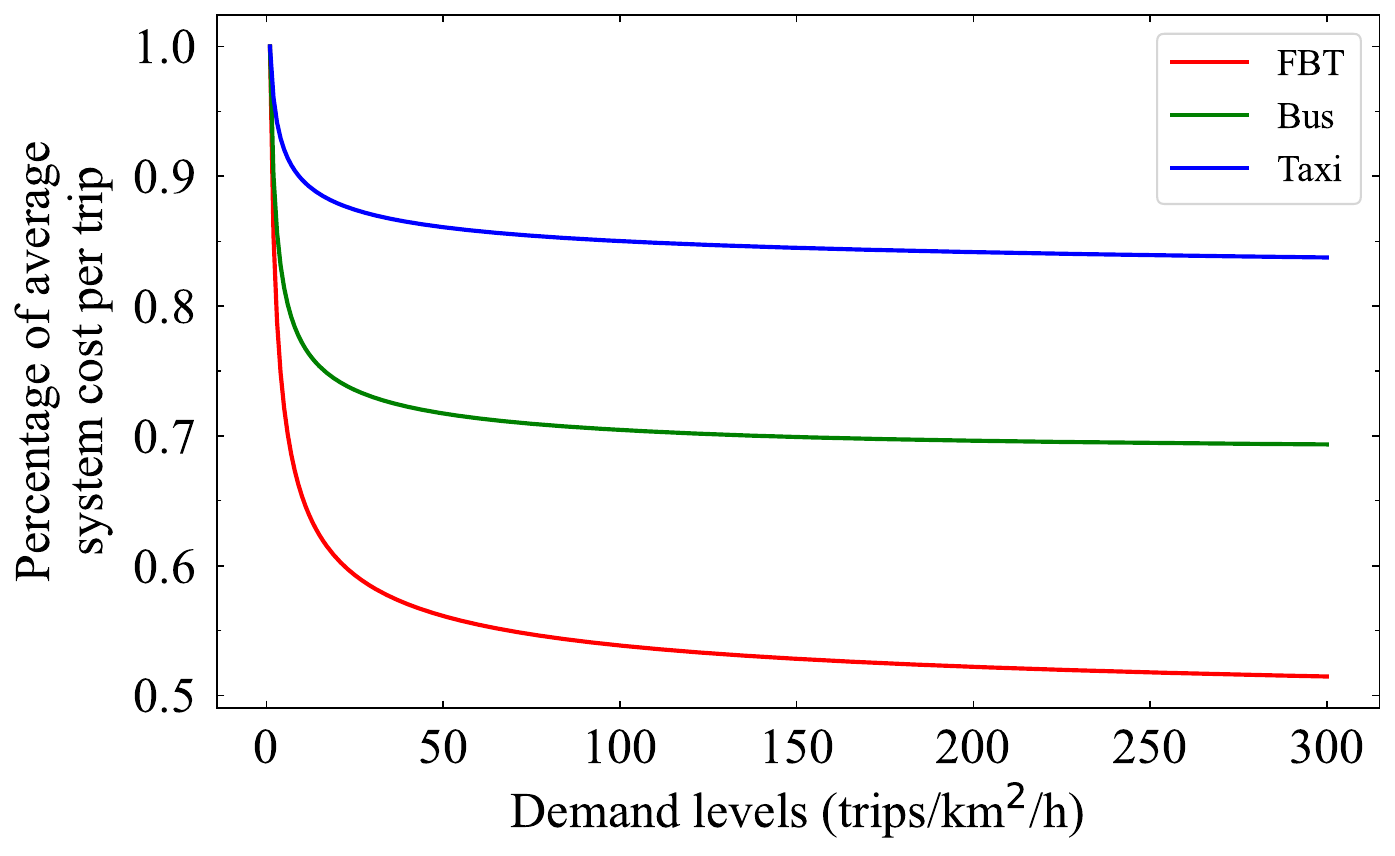}}
	\caption{Economies of scale (EoS) of FBT, taxi, bus systems.}
	\label{fig_EoS}
\end{figure}
\begin{figure}[htb]
	\centering
	\subfigure[High-wage cities $\alpha = \$25$/hr]{\includegraphics[width=75mm]{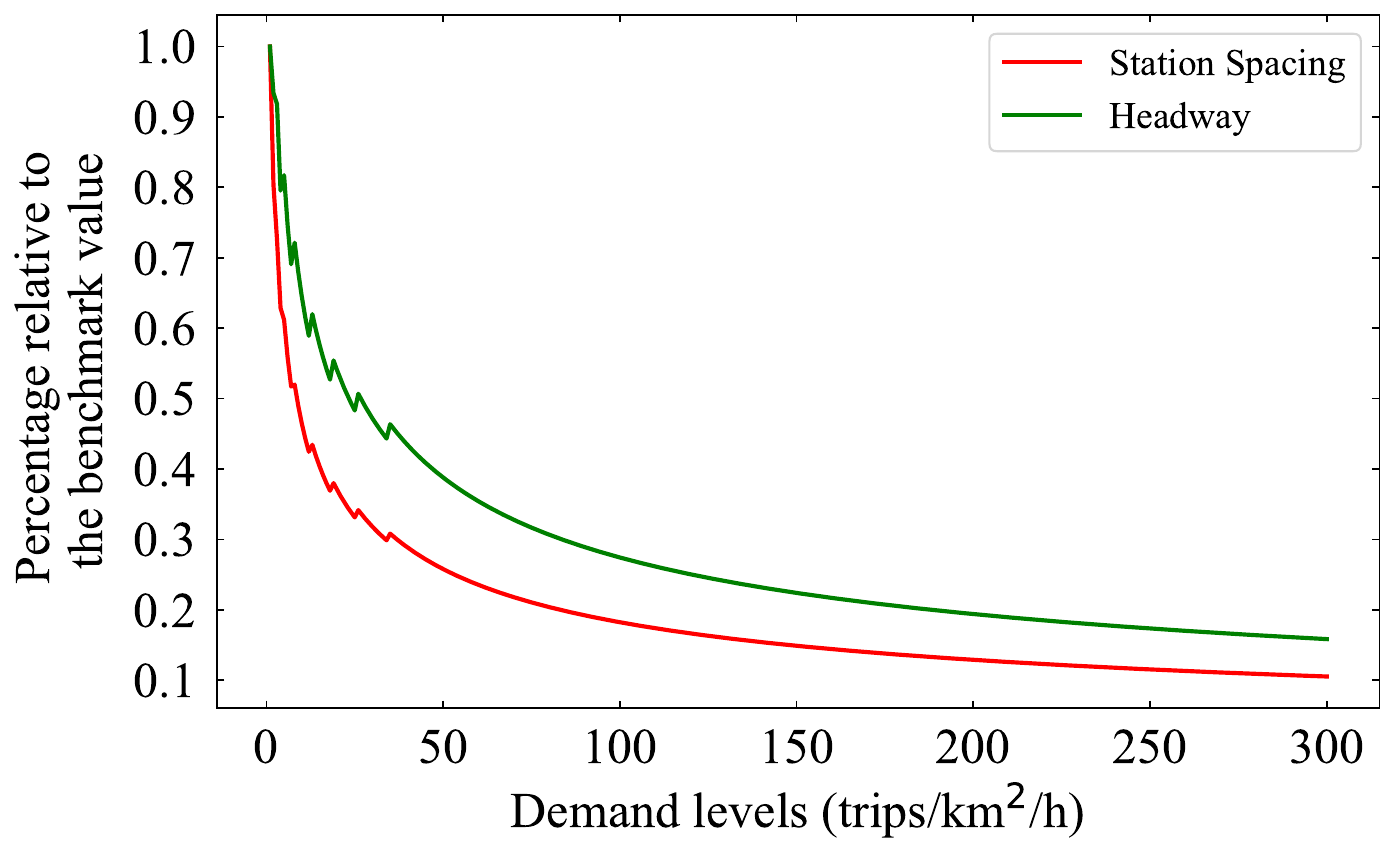}}
	\hfill
	\subfigure[Low-wage cities $\alpha = \$5$/hr]{\includegraphics[width=75mm]{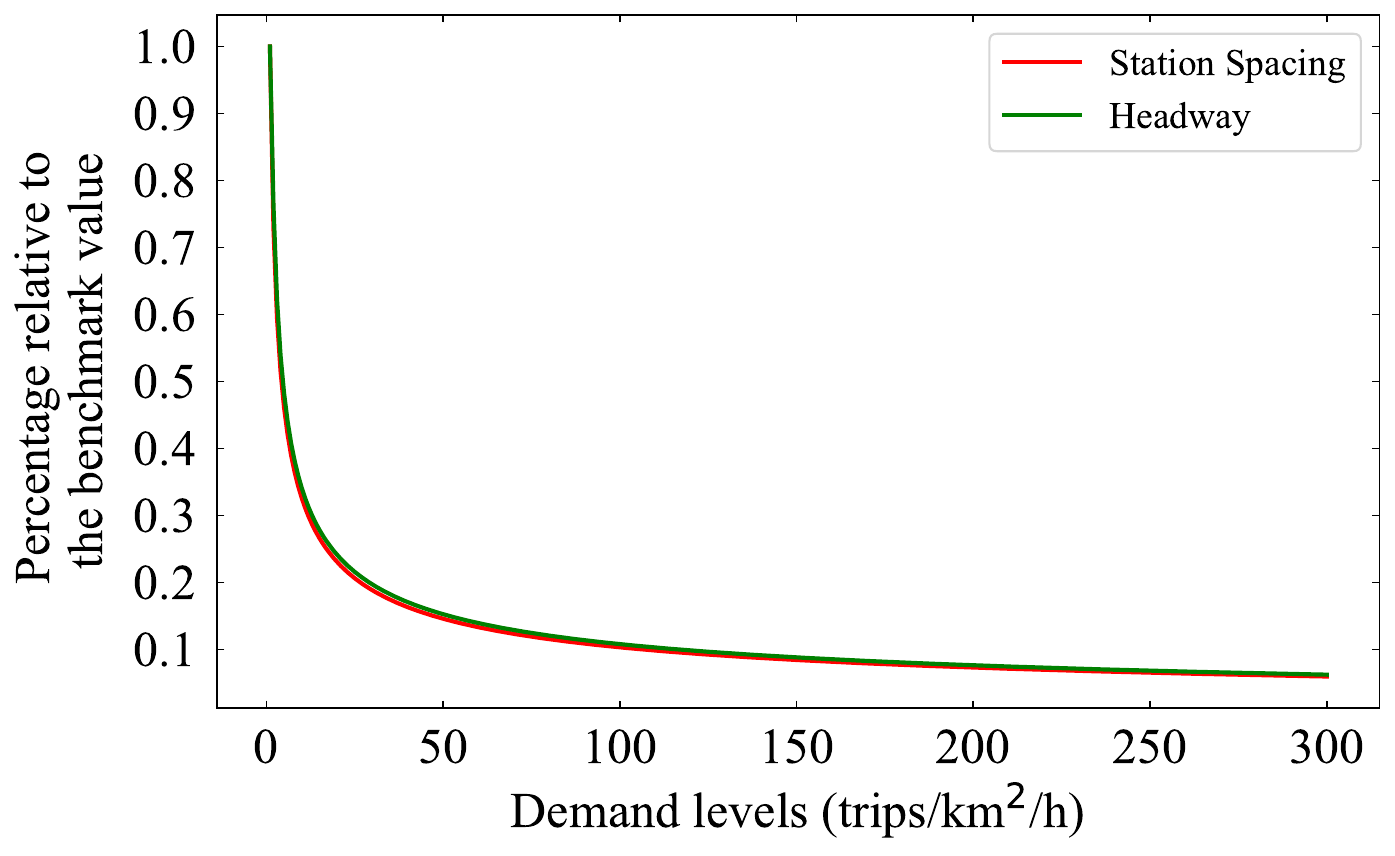}}
	\caption{Two sources of EoS in FBT systems: Decreasing headways and station spacings. (The sawtooth pattern in the curves results from the discrete integer values taken by $J^*$)}
	\label{fig_EoS_explanation}
\end{figure}

\subsubsection{System cost savings}
The average system costs of three optimized systems are depicted in Fig. \ref{fig_system_cost}(a) for high-wage cities and Fig. \ref{fig_system_cost}(b) for low-wage cities. As observed in Fig. \ref{fig_system_cost}(a), taxis prevail for low demand levels ($\lambda \leq 15$ trips/km$^2$/hr), producing the lowest average system cost. As demand levels rise, however, FBT outperforms taxis, while bus transit consistently incurs the highest average system cost. FBT's advantage is enhanced in low-wage cities, as shown in Fig. \ref{fig_system_cost}(b), where FBT triumphs across all demand levels in the tested scenarios.
\begin{figure}[htb]
	\centering
	\subfigure[High-wage cities $\alpha = \$25$/hr]{\includegraphics[width=75mm]{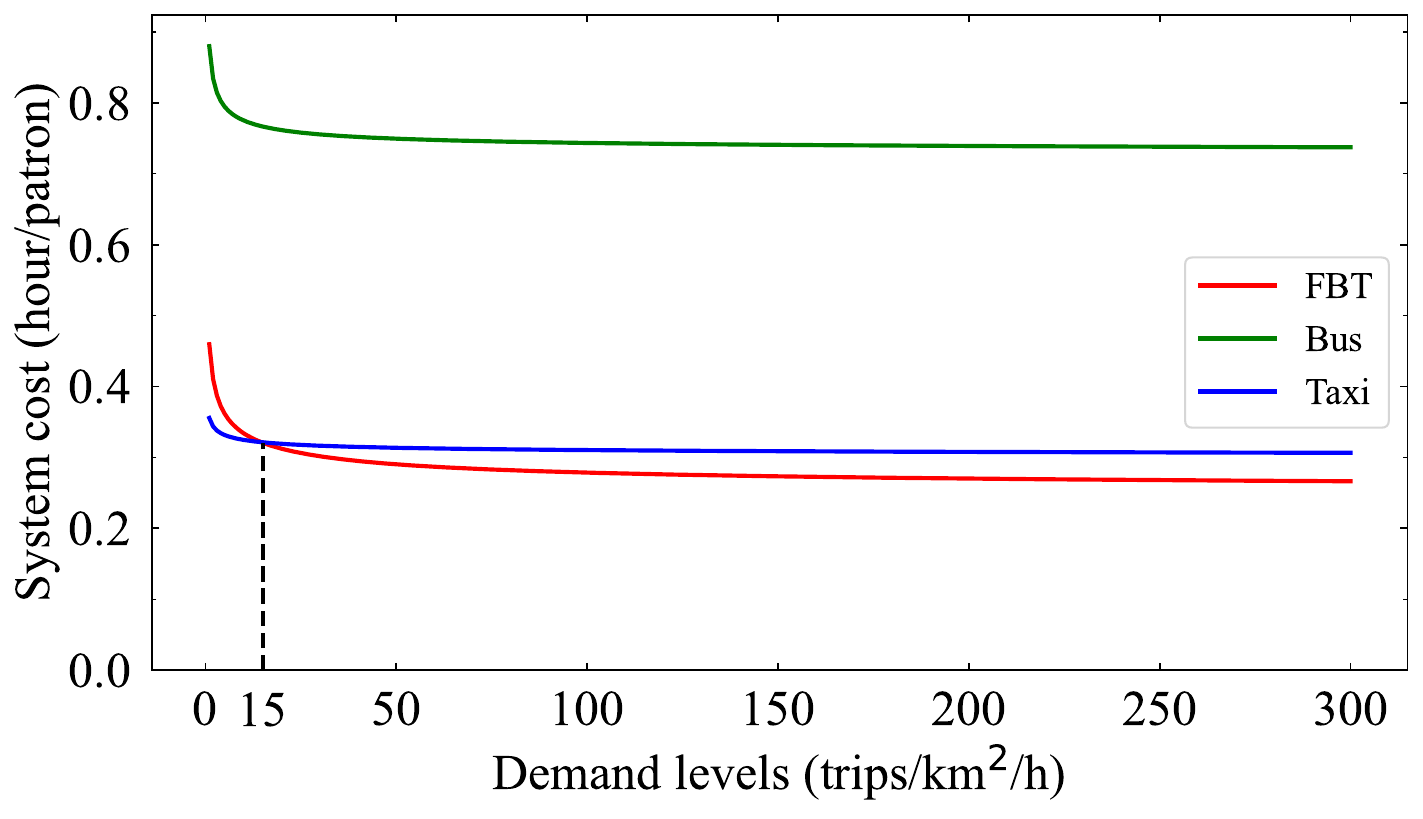}}
	\hfill
	\subfigure[Low-wage cities $\alpha = \$5$/hr]{\includegraphics[width=75mm]{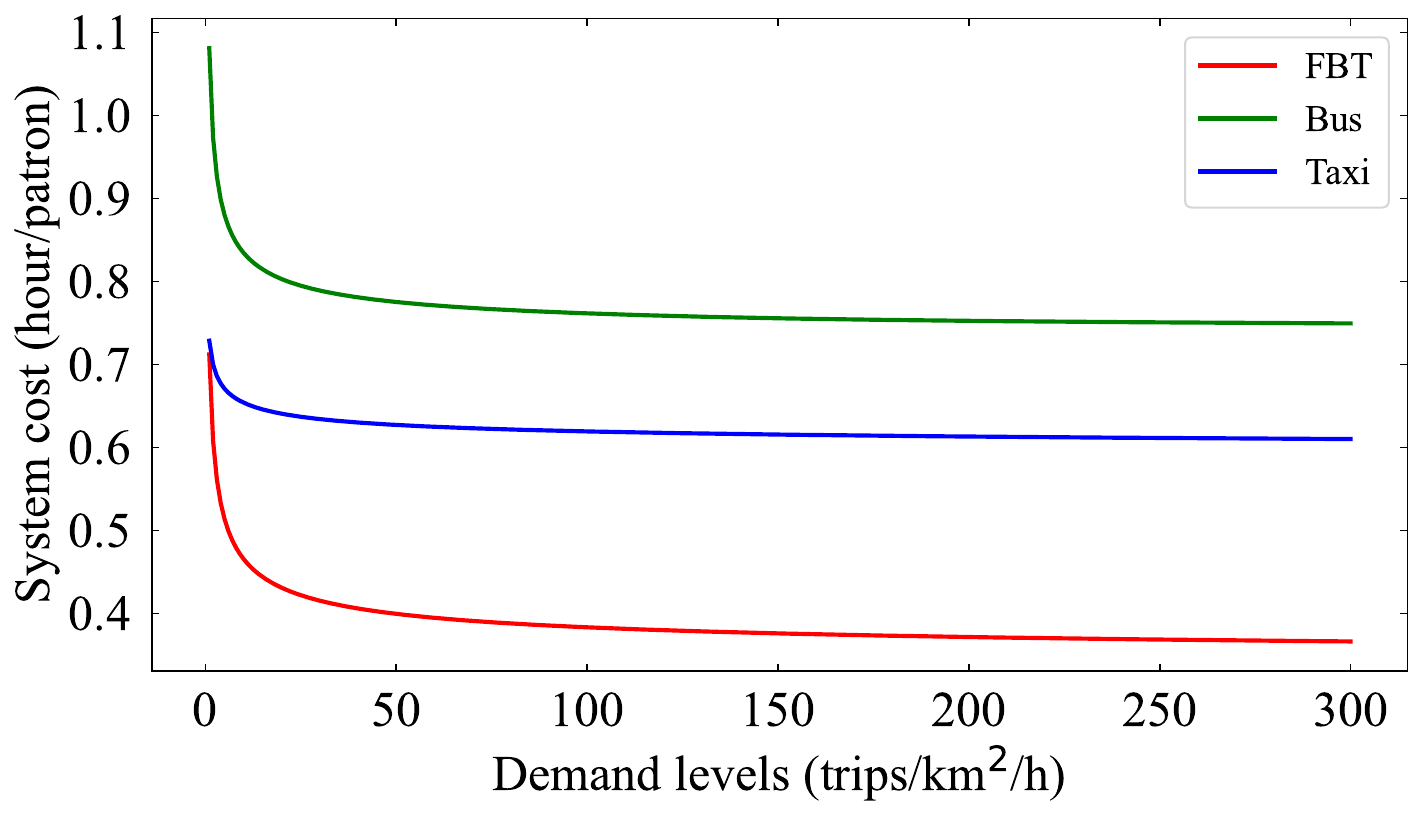}}
	\caption{Changes in average system cost of FBT, taxi, bus systems under optimized designs.}
	\label{fig_system_cost}
\end{figure}

Fig. \ref{fig_system_cost_saving}(a) and (b) present the cost savings of FBT relative to taxis and buses. FBT achieves savings of 13\% and 40\% compared to taxis in high- and low-wage cities, respectively, and 64\% and 51\% compared to buses in high- and low-wage cities, respectively.
\begin{figure}[htb]
	\centering
	\subfigure[High-wage cities $\alpha = \$25$/hr]{\includegraphics[width=75mm]{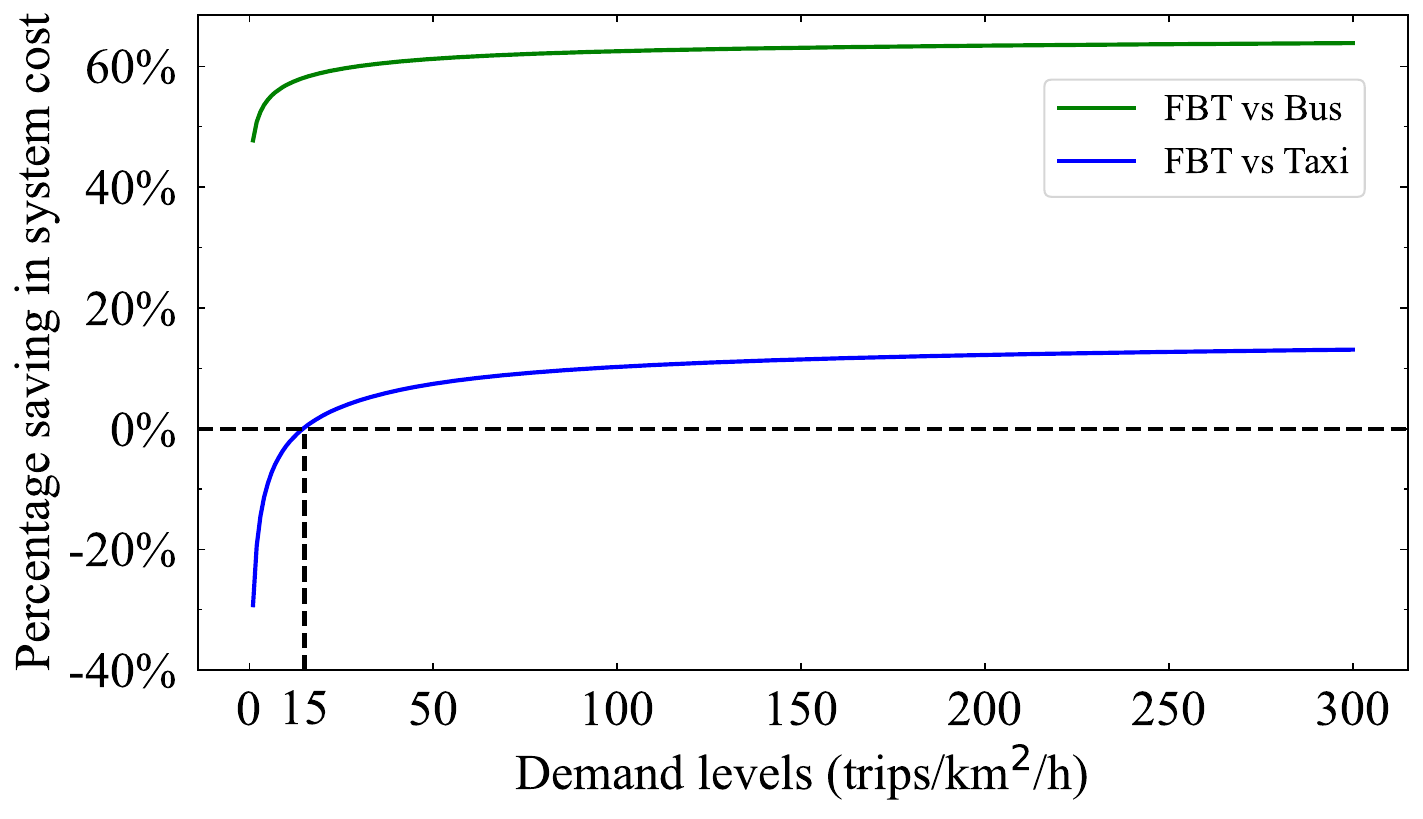}}
	\hfill
	\subfigure[Low-wage cities $\alpha = \$5$/hr]{\includegraphics[width=75mm]{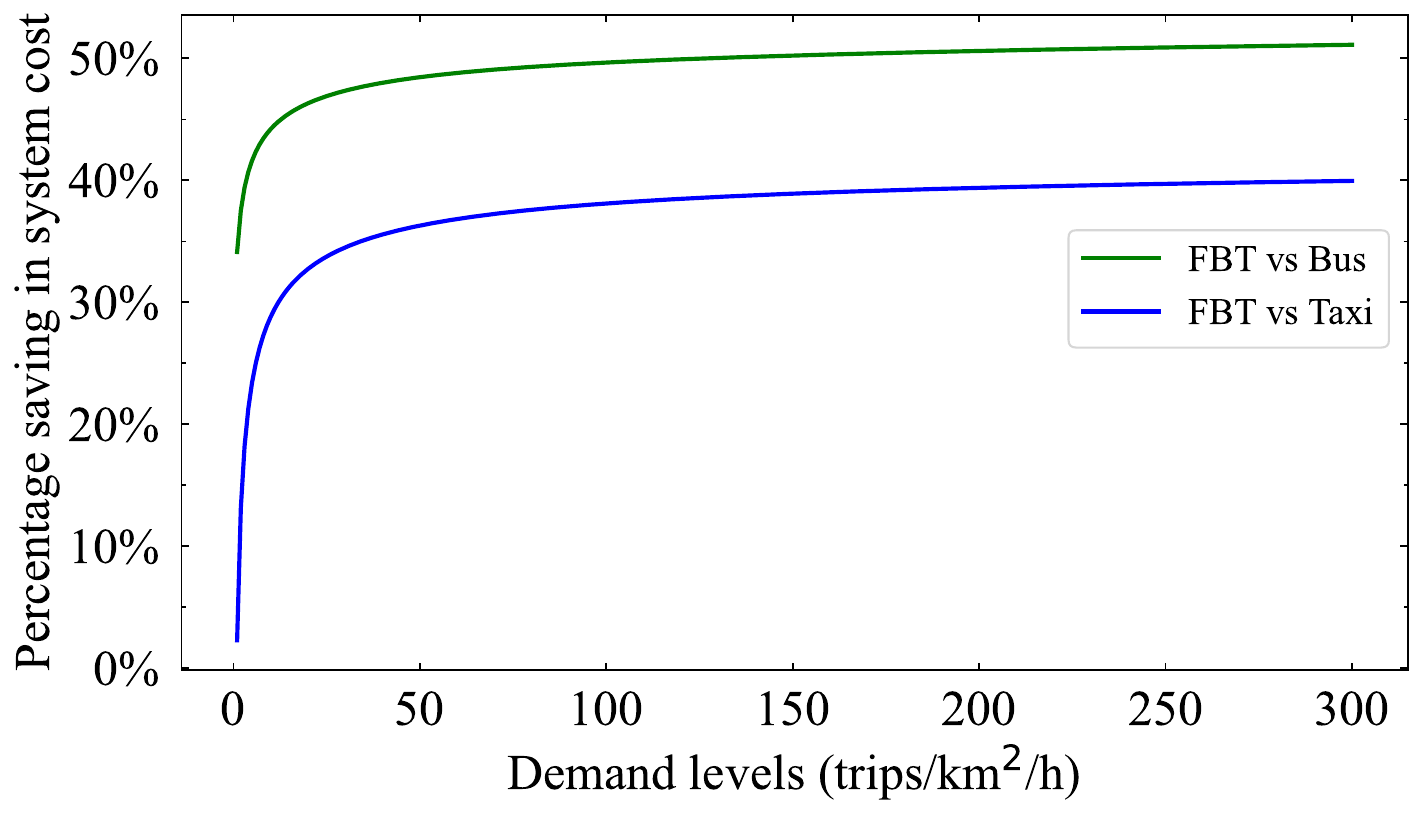}}
	\caption{System cost savings of FBT against taxi and bus systems.}
	\label{fig_system_cost_saving}
\end{figure}

\subsubsection{Parametric analysis}
The results above are determined by FBT’s unit cost parameters, particularly the capital and operational costs of trailers $\{ \pi_{\mathrm{tf}}, \pi_{\mathrm{tk}} \}$. To provide a conservative assessment, this section analyzes FBT’s performance with $\pi_{\mathrm{tf}}$ and $\pi_{\mathrm{tk}}$ scaled to 2, 4, and 6 times their baseline values. Notably, the scenarios with costs scaled by a factor of 6 represent the most conservative case, in which each trailer---despite being much smaller in size and battery---incurs the same cost as an all-purpose taxi. Even under such a conservative setting, FBT achieves up to 4\% and 18\% cost savings against taxis in high- and low-wage cities, respectively, owing to its strong EoS; see Fig. \ref{fig_system_cost_saving_var_cost}(e) and (f).
\begin{figure}[!htb]
	\centering
	\subfigure[$2\times \{ \pi_{\mathrm{tf}}, \pi_{\mathrm{tk}} \} \text{ in high-wage cities } \alpha = \$25$/hr]{\includegraphics[width=75mm]{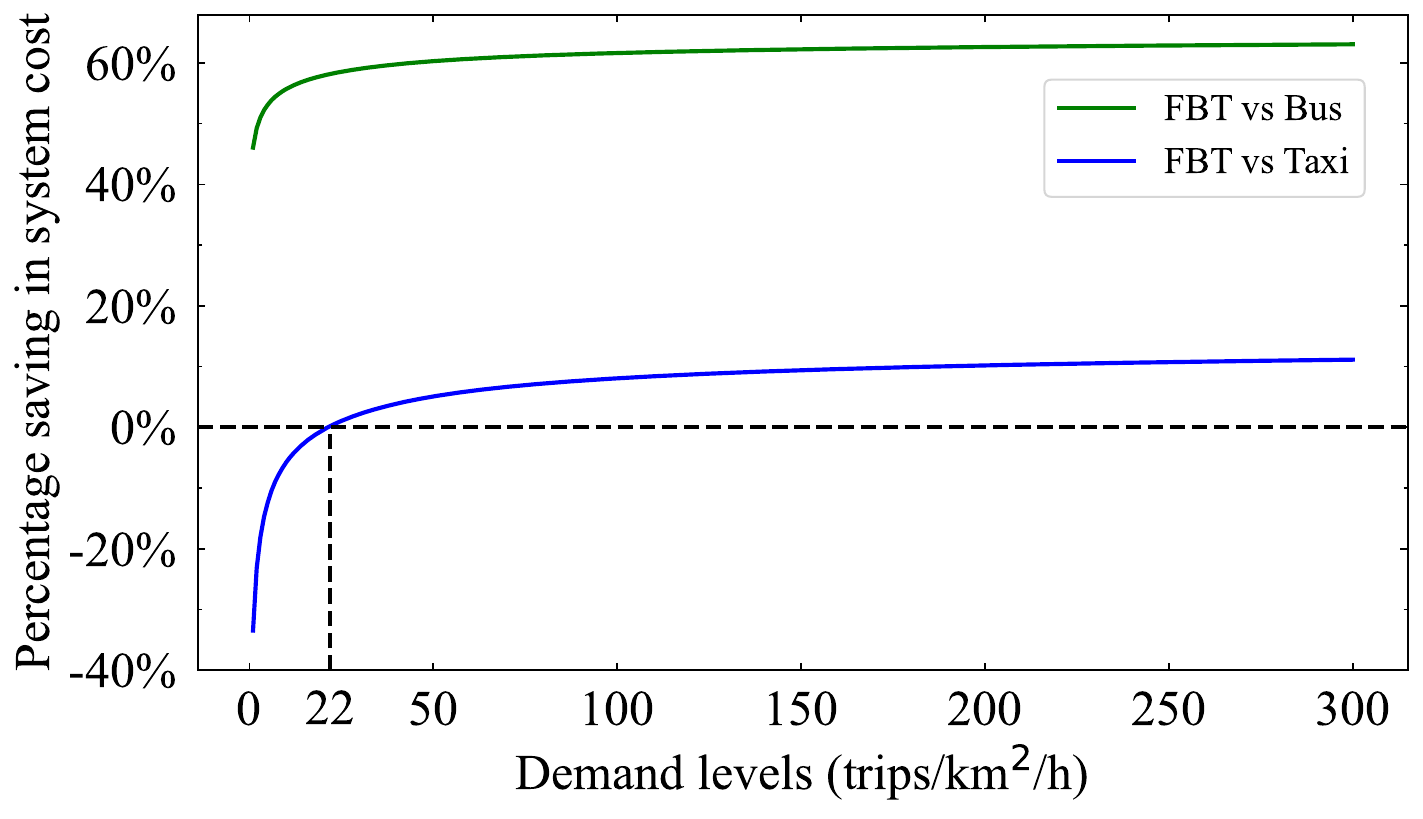}}
	\hfill
	\subfigure[$2\times \{ \pi_{\mathrm{tf}}, \pi_{\mathrm{tk}} \} \text{ in low-wage cities } \alpha = \$5$/hr]{\includegraphics[width=75mm]{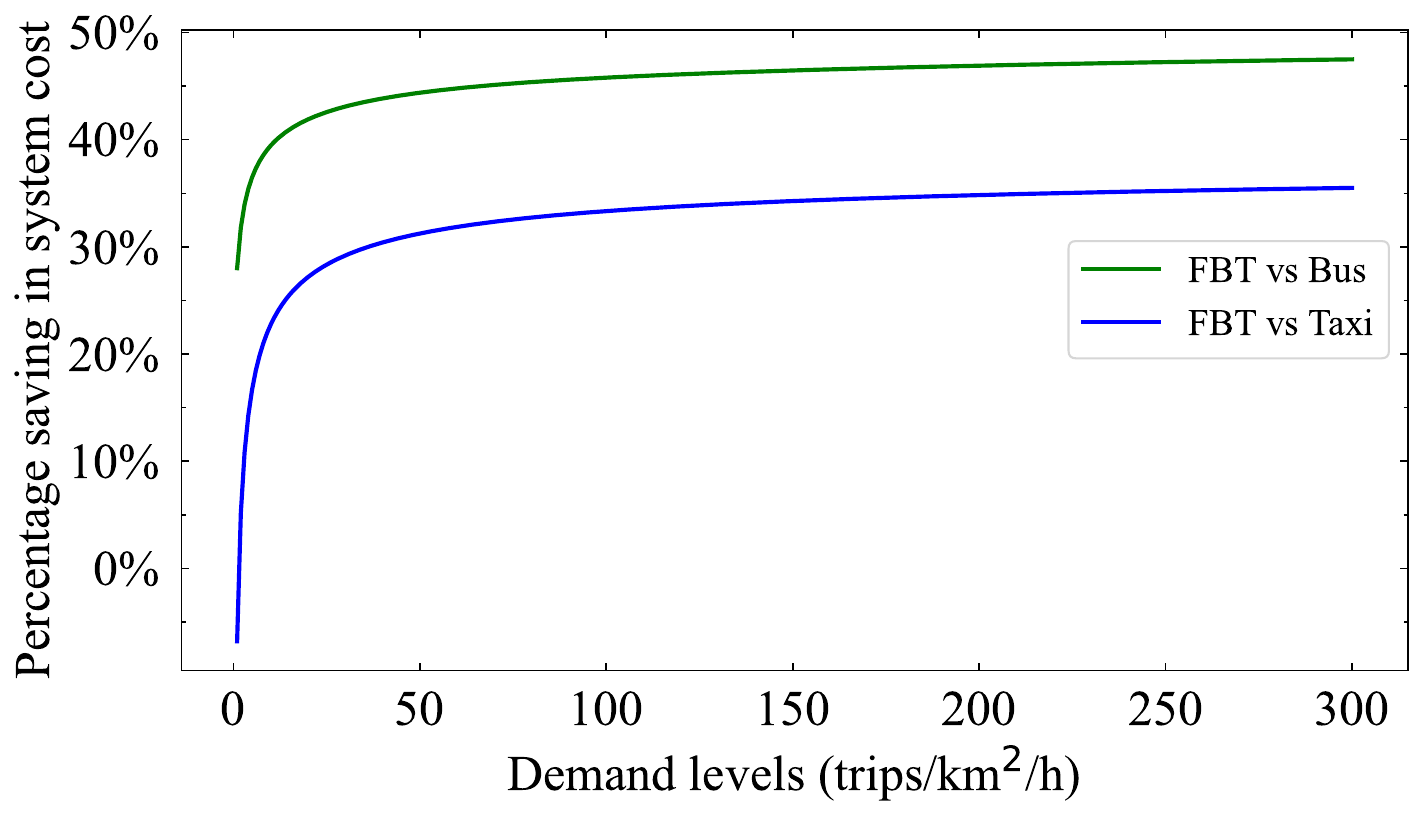}}
	\hfill
	\subfigure[$4\times \{ \pi_{\mathrm{tf}}, \pi_{\mathrm{tk}} \} \text{ in high-wage cities } \alpha = \$25$/hr]{\includegraphics[width=75mm]{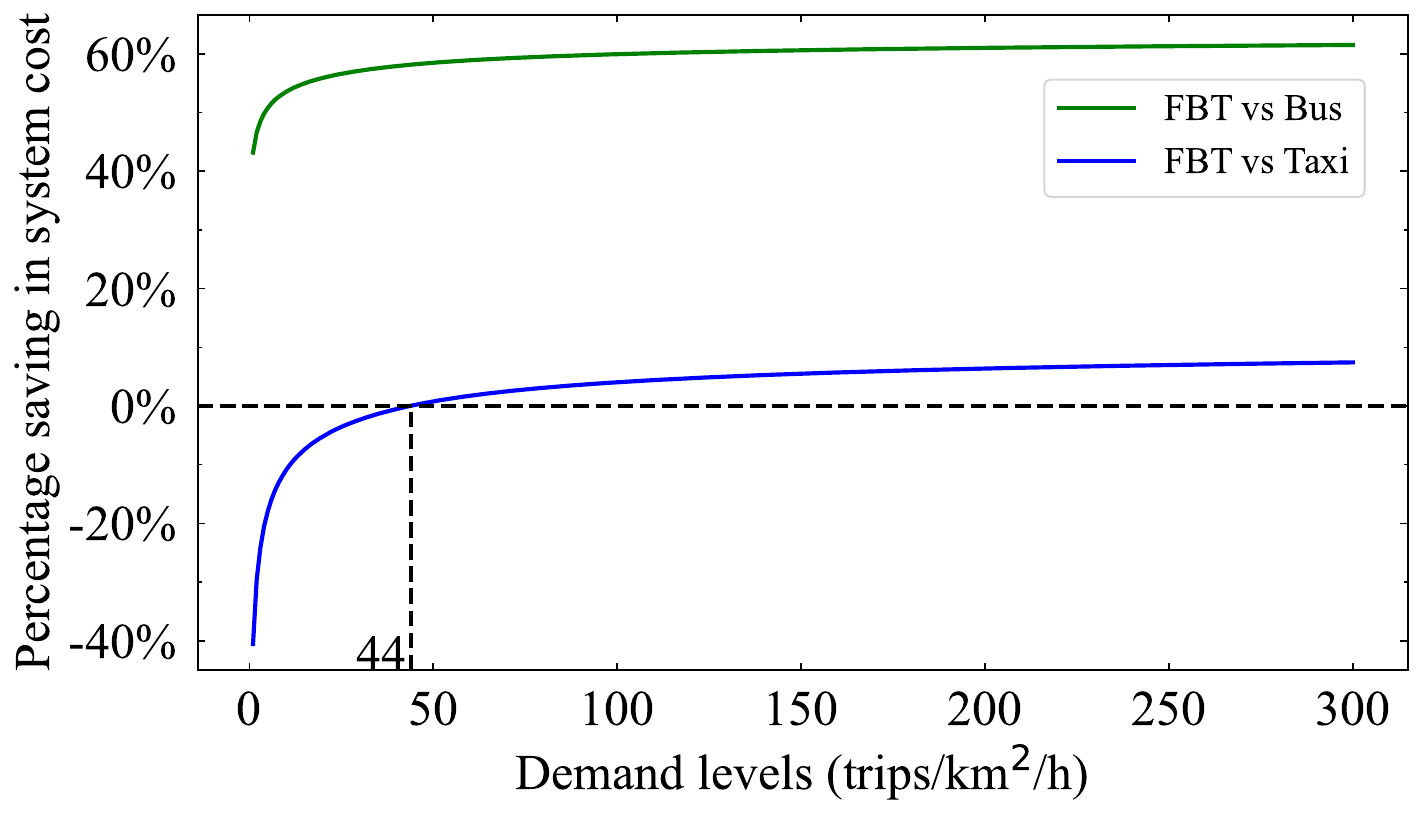}}
	\hfill
	\subfigure[$4\times \{ \pi_{\mathrm{tf}}, \pi_{\mathrm{tk}} \} \text{ in low-wage cities } \alpha = \$5$/hr]{\includegraphics[width=75mm]{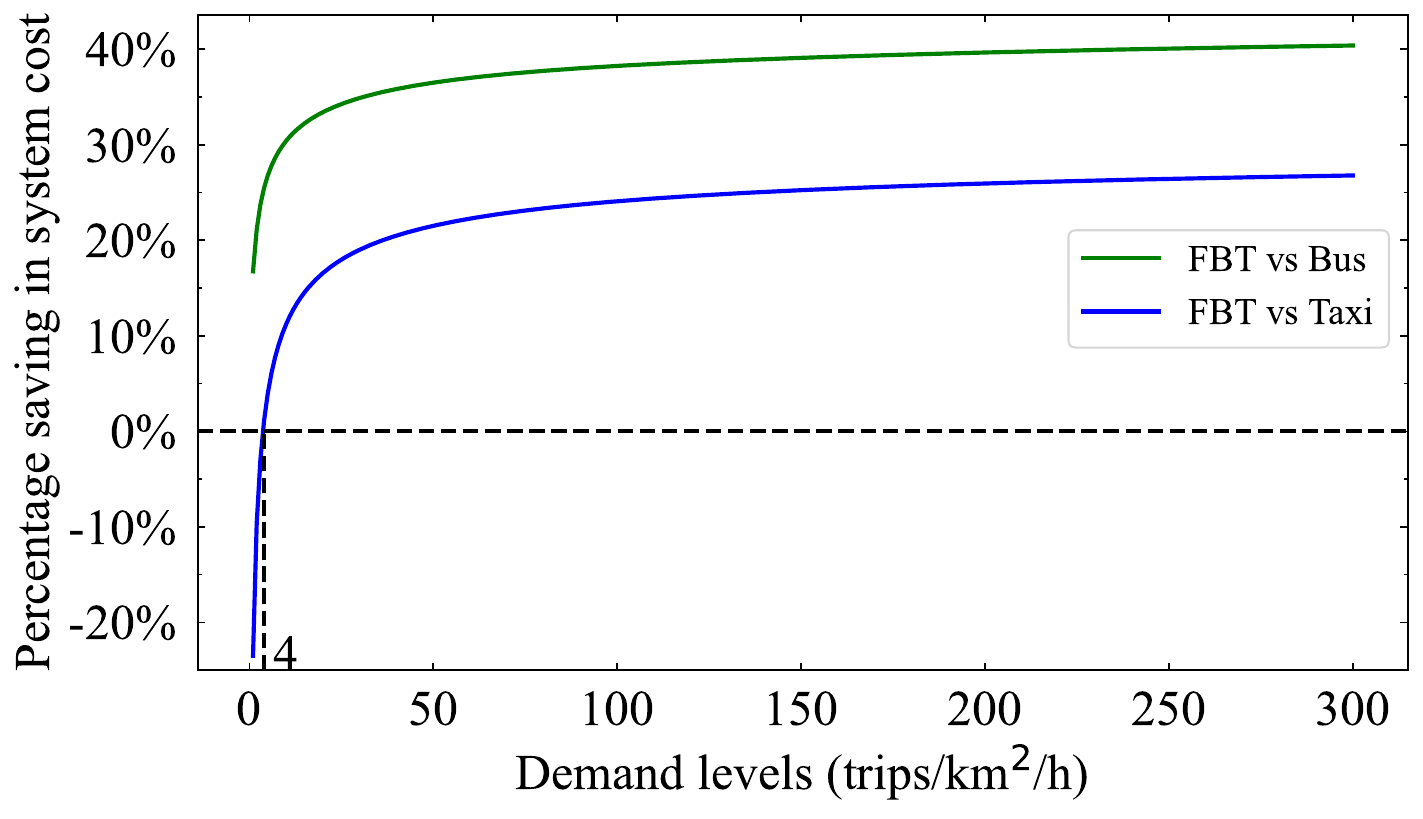}}
	\hfill
	\subfigure[$6\times \{ \pi_{\mathrm{tf}}, \pi_{\mathrm{tk}} \} \text{ in high-wage cities } \alpha = \$25$/hr]{\includegraphics[width=75mm]{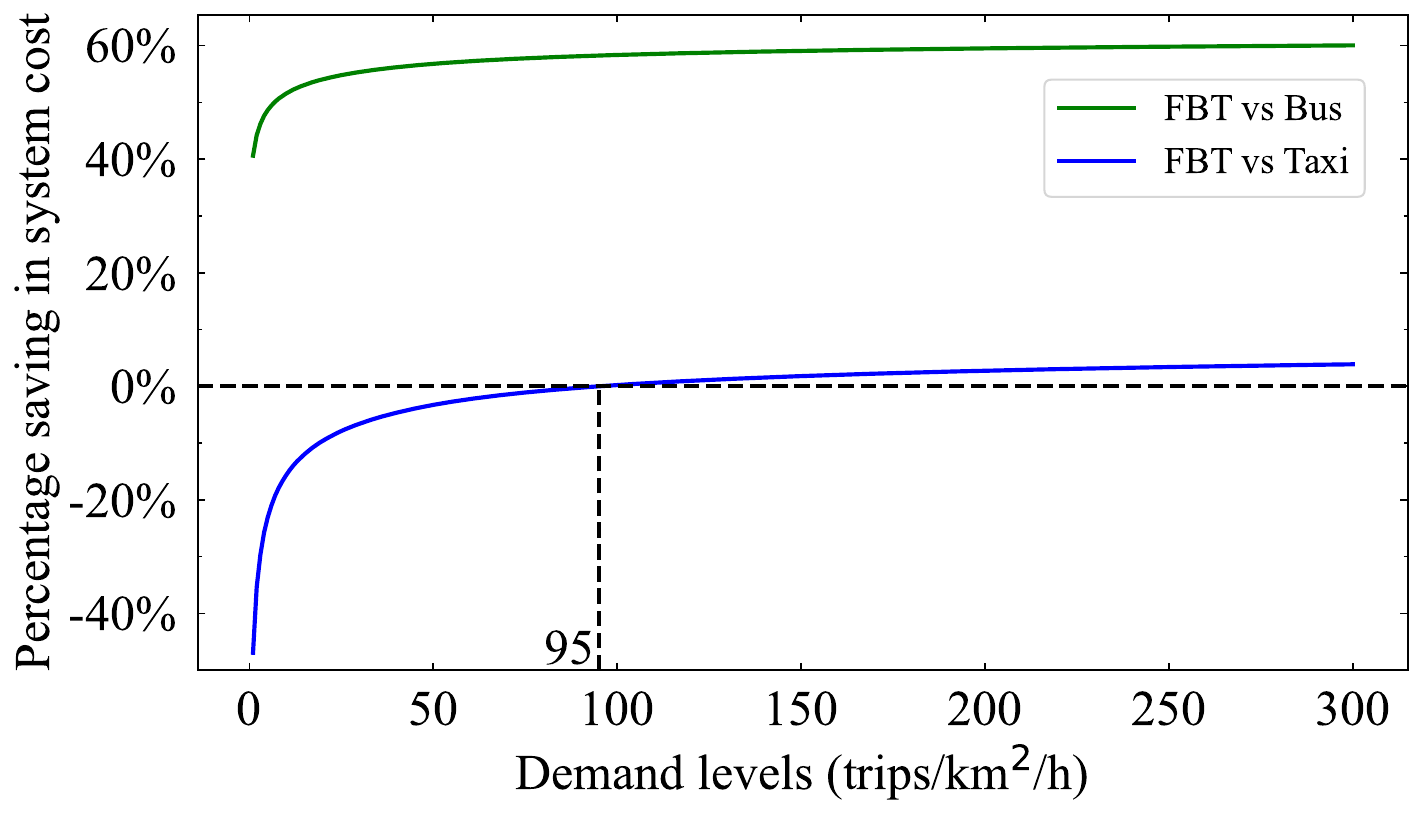}}
	\hfill
	\subfigure[$6\times \{ \pi_{\mathrm{tf}}, \pi_{\mathrm{tk}} \} \text{ in low-wage cities } \alpha = \$5$/hr]{\includegraphics[width=75mm]{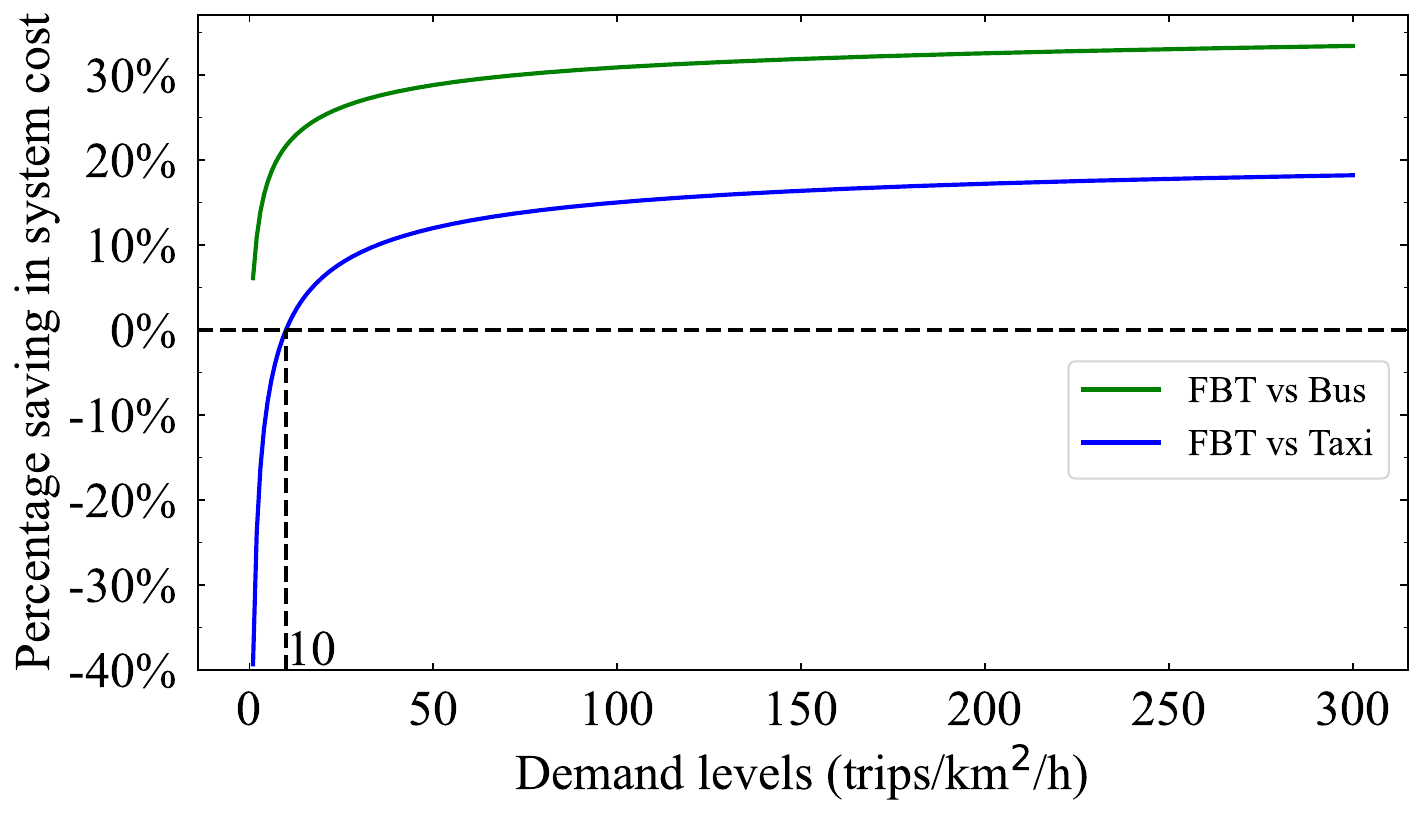}}
	\caption{Parametric analysis on unit cost settings.}
	\label{fig_system_cost_saving_var_cost}
\end{figure}

\section{Future research directions} \label{sec_agenda}
While the proof of concept above demonstrates the potential of FBT, transitioning from theoretical modeling to practical deployment requires a comprehensive research roadmap. This section proposes essential and promising future research directions, organized within a top-down research framework designed to advance foundational theories and algorithms for FBT. As illustrated in Fig. \ref{fig_agenda}, the framework is structured into five primary domains, each addressing a critical dimension of FBT strategic planning, transition planning, network design, simulation and operation, and real-world implementation. They are elaborated as follows. 

\begin{figure}[!ht]
	\centering
	\includegraphics[width=0.9\linewidth]{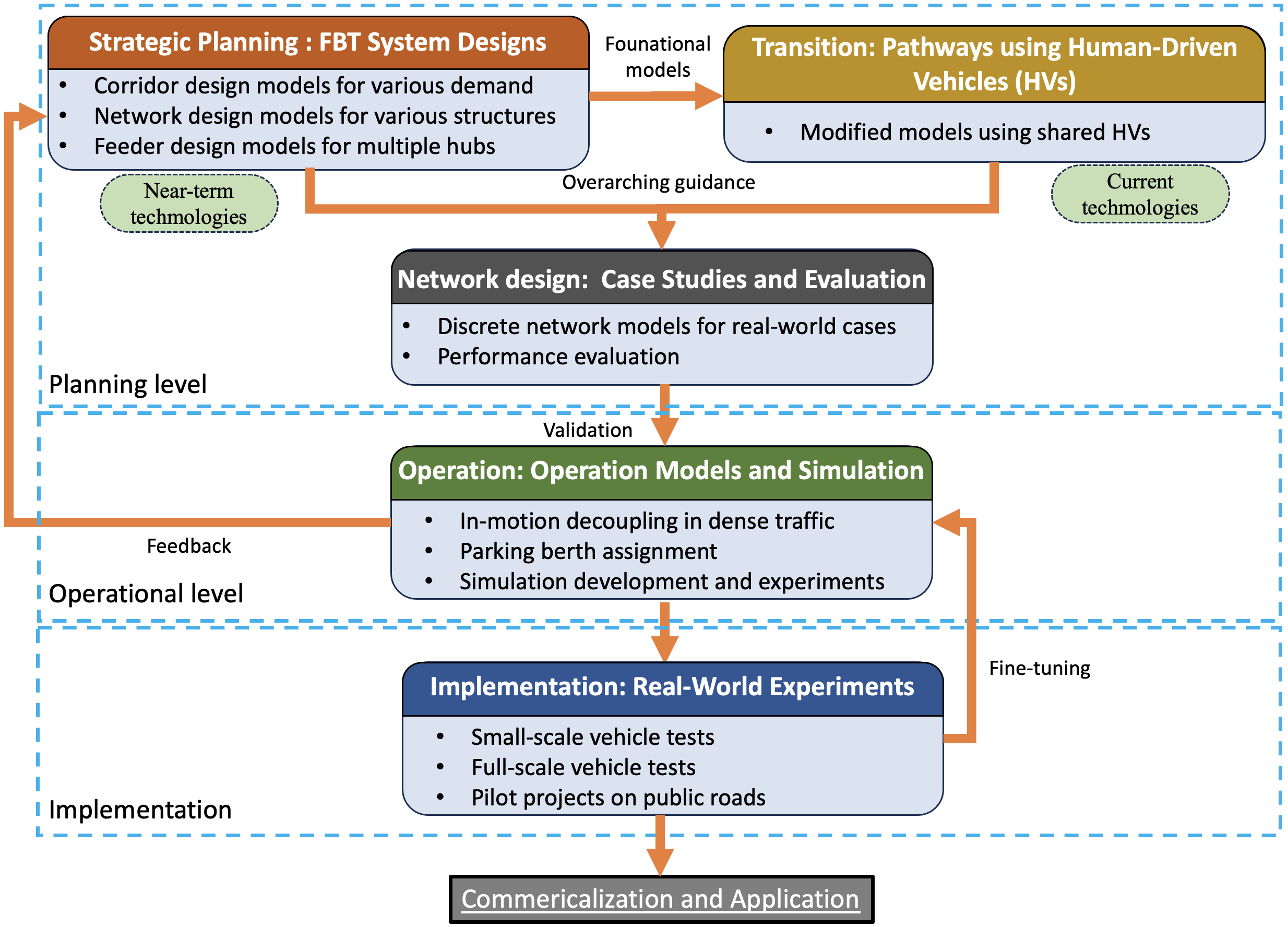}
	\caption{Future research for the development of FBT.}
	\label{fig_agenda}
\end{figure}

\subsection{Strategic planning for FBT} 
A promising first step involves extending the ring-shaped uniform corridor model to more general network structures, including heterogeneous corridors, citywide networks, and feeders to mass transit systems such as rail systems, while accounting for various demand distributions. These models should address both greenfield designs and integration with existing transport systems. Building on the Continuum Approximation (CA) framework, which has been widely used in large-scale transit studies \citep[e.g.,][]{fan2018optimal, wu2020, su2019optimal}, it is promising to derive analytical design solutions that generalize beyond specific network topologies. Such models are expected to distill fundamental design principles---relating to station spacing, platoon formation frequency, and fleet allocation---that can guide robust planning decisions across diverse urban contexts. 

A key objective for the strategic planning research is to develop tractable yet generalizable models that strike a balance between analytical rigor and practical applicability. The theoretical insights generated in this research will provide the foundation for empirical validation through real-world case studies in the following research domains.

\subsection{Transition planning for FBT} 
Recognizing the current dominance of HVs, another essential research direction is to investigate transition pathways that bridge existing vehicle technologies with FBT operations to cultivate market adoption. A straightforward but promising approach is to modify the access and egress components of the corridor, network, and FBT-as-feeder models using existing car-sharing and bike-sharing frameworks \citep{luo_joint_2021, wu2020}, while retaining FBT platoons as the trunk-line service. 

In this hybrid setting, shared trailers are parked within FBT's feeder zones, with the densities and capacities of these parking lots treated as decision variables. Passengers walk to/from the closest parking lot to access/return shared trailers and drive to the nearest FBT station for coupling into platoons. While the platoon will be led by a human driver, passengers in the trailers enjoy an ``autonomous'' travel experience, allowing them to engage in non-driving activities during transit. We expect that this hybrid model will allow for the incremental introduction of FBT elements into existing systems, enabling gradual adoption and evaluation of their impacts before large-scale deployment. 

\subsection{Specific FBT network designs} 
The third research direction involves translating the above analytical models into practice by building discrete network design models that incorporate real-world street networks, origin–destination demand matrices, and specific operational constraints related to key factors such as station locations, platoon lengths, and total fleet sizes. Case studies should be conducted across networks of varying scales---from medium-sized districts to large metropolitan areas---to demonstrate proof of concept and assess applicability.

A comprehensive evaluation framework is critical to stress-test FBT performance under diverse demand levels (including elastic demand), spatial and temporal variations, socio-economic characteristics, and network structures. Comparative analyses benchmarking FBT against conventional modes (e.g., taxis, buses) and existing MAV concepts will be instrumental in establishing the system's competitive viability.

\subsection{Simulation and operational studies} 
The fourth research direction focuses on simulation and operational studies. 
It is invaluable to develop a dedicated simulation platform for evaluating the operations of optimized FBT networks under complex and stochastic traffic conditions. The platform should capture the complexities of high-density urban environments such as Hong Kong, where severe spatial constraints pose unique operational challenges. Its primary function is to assess the system-wide impacts of in-motion trailer decoupling and generate critical insights into the feasibility and scalability of FBT in dense metropolitan contexts. Among foundational technologies for FBT operations, two essential ones are worth mentioning and warrant inclusion within the simulation platform: (1) in-motion decoupling under dense traffic, and (2) parking berth assignment for efficient stationary coupling. 

The first research should incorporate the influence of surrounding traffic into the in-motion decoupling trajectory planning. This includes accounting for uncertainties such as the stochastic deceleration of the leading vehicle in a platoon and the headway distributions of vehicles in adjacent lanes.
The second one addresses the optimal allocation of parking berths at FBT stations to trailers that arrive sequentially. Each trailer’s desired trip length, drawn from a known distribution, is revealed only at the moment of assignment to passengers. The challenge lies in minimizing the repositioning of trailers already stationed while ensuring that, upon platoon formation, trailers are ordered by non-increasing trip length. 

\subsection{Real-world experiments}
The final research direction concentrates on bridging theoretical developments and simulation findings with practical implementation through a structured experimental program. Three progressive stages can be designed: small-scale vehicle tests under controlled conditions, full-scale vehicle tests in semi-controlled environments such as campuses or industrial test fields, and pilot projects on public roads---aimed at validating system feasibility, operational safety, and public readiness.

Findings from these experiments can be used to calibrate simulation models from the fourth domain, refine control algorithms, and generate empirical evidence supporting regulatory approval and public trust. Passenger feedback and behavioral observations will further inform design optimization, ultimately guiding the pathway toward scalable, citywide deployment of FBT.

\subsection{Miscellaneous studies} 
Beyond these five research domains outlined above, two additional dimensions are critical for the successful implementation of FBT systems: \textit{infrastructure design} and \textit{regulatory alignment}. On the infrastructure side, the deployment of FBT necessitates careful consideration of existing road layouts, intersection geometries, and, in particular, station designs to accommodate stationary coupling and in-motion decoupling operations. Proactive infrastructure planning is essential to ensure operational safety, minimize disruptions to conventional traffic, and maximize the efficiency of FBT services. 

Equally important are the regulatory and institutional frameworks governing the introduction of novel vehicle types and operational models. The homologation process requires that FBT vehicles comply with all relevant safety, environmental, and operational standards before they can be legally deployed on public roads. To facilitate the smooth integration of FBT into urban mobility systems and to foster public acceptance, future research should systematically identify and address these regulatory and legal challenges.

Collectively, these efforts will form a comprehensive framework for evaluating FBT’s technical, operational, and socioeconomic viability as a novel mobility, informing future urban transportation strategies toward greater efficiency, sustainability, and accessibility.

\section{Conclusions}
This study has introduced fly-by transit (FBT) as a novel, near-term feasible modular transit system designed to deliver door-to-door, stop-less shared mobility. By integrating small-battery trailer modules for local feeder trips with high-performance leader modules for high-speed trunk-line platooning, FBT overcomes key limitations of existing MAV concepts, notably the reliance on technologically immature in-motion coupling and interior passage transfers. Instead, FBT employs stationary coupling and in-motion decoupling of tail trailers, enabling uninterrupted convoy movement while preserving technical practicality. 

As a proof of concept, optimal corridor design models are developed to demonstrate that FBT can outperform conventional buses in terms of travel time and taxis in terms of cost efficiency, thereby filling a critical gap in the spectrum of urban mobility options. In addition to conceptual and operational innovations, this study has outlined a research agenda that encompasses analytical system design, transition strategies using human-driven vehicles, comprehensive case studies, foundational operational development, and real-world experiments. Together, these efforts aim to establish the theoretical, operational, and economic basis for large-scale FBT deployment. 

By aligning with realistic technological capabilities and prioritizing patron convenience, FBT represents a pragmatic yet transformative step toward next-generation transportation systems. Its potential to reconcile accessibility, mobility, and cost-effectiveness positions it as a promising solution for sustainable, efficient, and inclusive urban transport in the coming decades.

Lastly, it is worth mentioning that FBT’s modular architecture also supports urban logistics. Trailer modules can be adapted for package movement, enabling seamless last-mile and middle-mile delivery. High-speed platooning on trunk lines facilitates rapid, reliable distribution of parcels, groceries, and other freight in dense urban environments. This dual-use capability enhances vehicle utilization and operational efficiency while meeting growing demand for integrated passenger–freight mobility. As cities seek to curb congestion and emissions from traditional delivery fleets, FBT provides a scalable, sustainable logistics alternative, strengthening its value proposition as a comprehensive mobility platform.

\section*{Appendix}

\appendix
\renewcommand{\theequation}{\Alph{section}\arabic{equation}}

\section{FBT trailer cost estimation} \label{appen_trailer_cost}
To estimate the cost of an FBT trailer, we benchmark it against mini-EVs on the market\footnote{\url{https://en.wikipedia.org/wiki/Wuling_Hongguang_Mini_EV}}, which share similar vehicle designs but are equipped with larger batteries (120–170 km range, with an average of 150 km) and more powerful motors capable of speeds over 100 km/h. We assume the only cost differential arises from the battery, which constitutes approximately 50\% of the mini-EV cost \citep{bobylev_wuling_2023, wevj12010021}, denoted as $\pi_{\text{m-EV}}$. Accordingly, the cost of an FBT trailer can be expressed as
\begin{align}
	\pi_{\mathrm{tf}}= \frac{1}{2}\pi_{\text{m-EV}} + \frac{1}{2}\pi_{\text{m-EV}}\frac{50}{150} = \frac{2}{3}\pi_{\text{m-EV}}.
\end{align}

In the current EV market, a standard electric taxi is priced at about four mini-EVs, i.e., $4\pi_{\text{m-EV}}$. This equivalence implies that the cost of one electric taxi is approximately equal to that of six FBT trailers.

\section{Expected travel distance of the leader} \label{appen_leader_distance}
\setcounter{equation}{0}
Given that the patron trip distance $\mathnormal{\ell}$ follows a uniform distribution, $\mathnormal{\ell} \sim \mathrm{U}(a, b)$, and the platoon size is $J$ (i.e., the number of patrons/trailers per dispatch), let $\mathcal{L}_J = \max_{j \in J}\{\ell_j\}$ denote the maximum trip distance among these patrons. The cumulative distribution function (CDF) of $\mathcal{L}_J$ is given by:
\[
\mathrm{CDF}(\mathcal{L}) = P(\mathcal{L}_J \leq \mathcal{L}) = \prod_{j=1}^{J} P(\ell_j \leq \mathcal{L}) = \left(\frac{\mathcal{L} - a}{b - a}\right)^{J},
\]
and the probability density function (PDF) is:
\[
\mathrm{PDF}(\mathcal{L}) = \frac{\mathrm{d}}{\mathrm{d} \mathcal{L}} \mathrm{CDF}(\mathcal{L}) = \frac{J}{(b - a)^{J}} (\mathcal{L} - a)^{J-1}.
\]
Therefore, the expected value of the maximum patron trip distance is derived as follows:
\begin{align}
	E( \mathcal{L}) &=\int _{a}^{b} \mathcal{L} \mathrm{PDF}( \mathcal{L} ) \mathrm{d}\mathcal{L} \notag\\
	&=\int _{a}^{b} \mathcal{L} \cdotp \frac{J}{( b-a)^{J}}( \mathcal{L} -a)^{J-1} \mathrm{d}\mathcal{L} \notag\\
	(\text{Let } x=\mathcal{L} -a) &=\int _{0}^{b-a}( x+a) \cdotp \frac{J}{( b-a)^{J}} x^{J-1} \mathrm{d}x\notag\\
	&=\frac{J}{( b-a)^{J}}\int _{0}^{b-a}\left( x^{J} +ax^{J-1}\right) \mathrm{d}x\notag\\
	&=\frac{J}{( b-a)^{J}} \cdotp \left(\frac{( b-a)^{J+1}}{J+1} +a\frac{( b-a)^{J}}{J}\right)\notag\\
	&=\frac{J( b-a)}{J+1} +a\notag\\
	&=\frac{Jb+a}{J+1}.
\end{align}

\section{Optimal design models for traditional bus and taxi systems} \label{appen_taxi_bus}
\setcounter{equation}{0}
For completeness, we briefly summarize the key components of the optimal design models for taxis and buses, without providing detailed derivations or explanations. For further information, readers are referred to \cite{daganzo2019public}. To maintain clarity and limit symbol proliferation, we use similar notations to represent analogous metrics across different systems; they appear in respective sections to avoid confusion.
\subsection{Taxi model}
The optimal design of taxis involves a single decision variable, i.e., the density of idle taxis, $f_i$, as given below: 
\begin{subequations} \label{eq_taxi}
	\begin{align}
		\minimize_{f_i} Z(f_i) = \frac{U_o}{\alpha} + U_p,
	\end{align}
	subject to 
	\begin{align}
		f_i & > 0,
	\end{align}
	where the costs of the taxi operator and patrons are
	\begin{align}
		U_{o} & = \pi \mathrm{_{cf}}\left(\frac{f_{i}}{\lambda } +\frac{k}{v}\sqrt{\frac{1}{f_{i}}} +\frac{W}{2v} +\frac{a+b}{2V}\right) +\pi \mathrm{_{ck}}\left( k\sqrt{\frac{1}{f_{i}}} +\frac{W}{2} +\frac{a+b}{2}\right), \\
		U_{p} & =\frac{k}{v}\sqrt{\frac{1}{f_{i}}} +\frac{W}{4v} \times 2+\frac{a+b}{2V}, 
	\end{align}
\end{subequations}
and $\pi_{\mathrm{cf}}$ and $\pi_{\mathrm{ck}}$ are unit capital and operational costs of taxis. In the numerical experiments, we set $\pi_{\mathrm{cf}}=\$0.498$/taxi/hr  and $\pi_{\mathrm{ck}}=\$0.228$/taxi/km.

The closed-form solution to (\ref{eq_taxi}) is given by
\begin{align}
	f_{i}^{*}=\left(\frac{\lambda k\left(\pi\mathrm{_{cf}}+\pi\mathrm{_{ck}}v+\alpha\right)}{2\pi\mathrm{_{cf}}v}\right)^{\frac{2}{3}},
\end{align}
which leads to
\begin{align}
	Z^* = A\lambda^{-\frac{1}{3}} + B
\end{align}
where $A$ and $B$ are constants determined by pre-given parameters.

\subsection{Bus model}
The optimal design of bus transit includes two decision variables, stop spacing $S$ and service headway $H$, as given below:
\begin{subequations} \label{eq_bus}
	\begin{align}
		\minimize_{S,H} Z(S,H) = \frac{U_o}{\alpha} + U_p
	\end{align}
	subject to:
	\begin{align}
		\lambda W\frac{a+b}{2} H\leqslant C
	\end{align}
	where the costs of the bus operator and patrons are
	\begin{align}
		U_{o} & =\frac{1}{\lambda WH}\left(\frac{\pi_{\mathrm{bf}}}{V}+\frac{\pi_{\mathrm{bf}}t_{s}}{S}+\pi_{bk}\right)+\frac{\pi_{S}}{\lambda SW}, \label{eq_Uo_App} \\
		U_p & = \frac{W+S}{2v_{w}}+\beta\frac{H}{2}+\frac{a+b}{2V}+t_{s}\frac{a+b}{2S}, \label{eq_Up_App}
	\end{align}
\end{subequations}
and $\pi_{\mathrm{bf}}$ and $\pi_{\mathrm{bk}}$ are unit capital and operational costs of buses; $C$ means bus capacity; $v_w$ is patrons' average walking speed; and $\beta$ is a weight parameter on out-of-vehicle waiting time. In the numerical experiments, we set $\pi_{\mathrm{bf}}=\pi_{\mathrm{lf}}, \pi_{\mathrm{bk}}=\pi_{\mathrm{lk}}$, $C=45$ patrons/bus, $v_w = 5$ km/hr, and $\beta = 2$.

The optimal conditions for (\ref{eq_bus}) are
\begin{align} \label{eq_S_H}
	S^{*} &=\sqrt{v_{w}\left(\frac{2\pi _{\mathrm{bf}} t_{s}}{\lambda \alpha WH^* } +\frac{2\pi_{S}}{\lambda \alpha W } +t_{s}( a+b)\right)}, \\
	H^{*} &=\sqrt{\frac{2}{\lambda \alpha \beta W }\left(\frac{\pi _{\mathrm{bf}}}{V} +\frac{\pi _{\mathrm{bf}} t_{s}}{S^*} +\pi \mathrm{_{\mathrm{b} k}}\right)}.
\end{align}

With (\ref{eq_S_H}a and b) in hand, an iterative algorithm can be devised to efficiently solve (\ref{eq_bus}) to the global optimum. (Note that $U_o$ and $U_p$ in (\ref{eq_Uo_App}, \ref{eq_Up_App}) are posynomial functions of $S,H$, and can be transformed into convex functions through geometric programming \citep{boyd2004convex}; and thus, (\ref{eq_bus}) yields a unique global optimum.)

\begin{remark}
	It is interesting to consider an idealized scenario where operational costs are negligible ($\pi_{\mathrm{bf}}= \pi_{\mathrm{bk}}=\pi_{S}=0$) and service frequency is infinite ($H \rightarrow 0$). In this extreme case, the optimized stop spacing and the minimized patrons' travel time become
	\begin{align} \label{eq_S_H}
		S^{*} &=\sqrt{v_{w}t_{s}\left( a+b\right)}, \\
		U_{p}^{*} & = \sqrt{\frac{t_{s}\left(a+b\right)}{v_{w}}}+\frac{W}{2v_{w}}+\frac{a+b}{2V}.
	\end{align}
	
	The resulting optimized door-to-door speed for the bus system, given by $\frac{a+b+W}{2U_{p}^{*}}$,
	yields approximately 10 km/h under the parameters used in our numerical experiments, and increases to about 22 km/h for patrons located directly on the bus line ($W=0$). The findings, consistent with those in \cite{daganzo2019public}, imply that the bus service speed remains close to that of cycling, even under these ideal conditions.
\end{remark}

\section*{Replication and data sharing}
All data have been included in the main text of the paper.



%
\section*{Declaration of competing interest}

The authors declare that they have no known competing financial interests or personal relationships that could have appeared to influence the work reported in this paper.

\section*{Acknowledgment}
This work was supported by the Start-up Fund of The Hong Kong Polytechnic University and, in part, by the Otto Poon Charitable Foundation Smart Cities Research Institute (SCRI) at The Hong Kong Polytechnic University (Project P0058095). We are grateful to Prof. Yu (Macro) Nie and Prof. Edward Chung for their constructive comments on an earlier draft. In a private conversation, Prof. Yueyue Fan vividly recalled her childhood imagination of a similar mobility solution. The first author also thanks his students, Jie Chen and Zongjie Pan, for assistance in producing several figures and tables. An early version of this paper was presented at the 29th International Conference of Hong Kong Society for Transportation Studies.

%


\printcredits

\bibliographystyle{cas-model2-names}

\bibliography{cas-refs}



\end{document}